\RequirePackage{amsmath}
\documentclass[smallextended,runningheads]{svjour2}                    
\smartqed  
\usepackage{graphicx}
\usepackage{bm}
\usepackage{amssymb}
\usepackage{enumitem}
\usepackage{tikz}
\usepackage{cite}
\usepackage{amsfonts,amssymb,mathtools}
\mathtoolsset{showonlyrefs=true}
\newcommand{\bn}{\bm{n}}
\newcommand{\be}{\bm{e}}
\DeclareMathOperator\erf{erf}

\newcommand\restr[2]{{
  \left.\kern-\nulldelimiterspace 
  #1 
  \vphantom{\big|} 
  \right|_{#2} 
}}%

\begin{document}

\title{Motion of a Rigid Body in a Special Lorentz Gas: Loss of Memory Effect}

\titlerunning{Motion of a Body in a Special Lorentz Gas}

\author{Kai Koike}

\authorrunning{Kai Koike} 

\institute{Kai Koike \at
1) School of Fundamental Science and Technology, Keio University, 3-14-1 Hiyoshi, Kohoku-ku, Yokohama 223-8522, Japan \\
\email{koike@math.keio.ac.jp} \\
2) Mathematical Science Team, RIKEN Center for Advanced Intelligence Project, 1-4-1 Nihonbashi, Chuo-ku, Tokyo 103-0027, Japan
}

\date{Received: date / Accepted: date}

\maketitle

\begin{abstract}
  Linear motion of a rigid body in a special kind of Lorentz gas is mathematically analyzed. The rigid body moves against gas drag according to Newton's equation. The gas model is a special Lorentz gas consisting of gas molecules and background obstacles, which was introduced in~(Tsuji and Aoki: J. Stat. Phys. \textbf{146}, 620--645, 2012). The specular boundary condition is imposed on the resulting kinetic equation. This study complements the numerical study by Tsuji and Aoki cited above --- although the setting in this paper is slightly different from theirs, qualitatively the same asymptotic behavior is proved: The velocity $V(t)$ of the rigid body decays exponentially if the obstacles undergo thermal motion; if the obstacles are motionless, then the velocity $V(t)$ decays algebraically with a rate $t^{-5}$ independent of the spatial dimension. This demonstrates the idea that interaction of the molecules with the background obstacles destroy the memory effect due to recollision.
\end{abstract}

\section{Introduction}
Fluid force acting on a moving body in a fluid is not solely determined by the instantaneous velocity of the body; it also depends on the past history of motion. This is because the disturbance made in the fluid is not immediately wiped away and affects the future motion of the body. The important of this non-Markovian nature --- \textit{memory effect} --- of fluid force is well-studied for viscous fluids~(see e.g.~\cite{Daitche2011,Daitche2015} and the references therein).

Analysis of force acting on a moving body in \textit{rarefied gas} has attracted attention due to its importance in MEMS~(micro-electro-mechanical systems) and vacuum technology~\cite{Russo2009,Shrestha2015,Shrestha2015a,Shrestha2015b,Versluis2009,Tiwari2013,Rader2011,Kobert2015,Dechriste2012,Jin2016,Tsuji2012,Tsuji2013,Tsuji2014,Tsuji2015}. Mathematical studies are relatively limited and only \textit{free molecular flow} (gas flow which is so dilute that collisions among molecules can be neglected) has been analyzed~\cite{Aoki2008,Aoki2009,Butta2015,Caprino2007,Caprino2006,Cavallaro2007,Cavallaro2012,Chen2014,Chen2015,Chen2015a,Fanelli2016,Koike2017,Sisti2014}.

Let us review in particular a result by Caprino et al.~\cite{Caprino2007}. Consider a linear motion of a cylindrical rigid body in a free molecular flow that was otherwise at rest~(see Fig.~\ref{fig:cylinder}). And assume that the force acting on the body is just the gas drag $D(t)$ and that molecules elastically reflect at the surface of the body (specular reflection). They proved that the velocity $V(t)$ of the moving body decays only algebraically as $V(t)\approx t^{-(d+2)}$, where $d$ is the spatial dimension (see Section~\ref{sec:formulation} for a more precise definition of $d$). This algebraic decay is caused by the non-Markovian nature of the drag $D(t)$ (i.e., $D(t)$ is not solely determined by $V(t)$), which shows that memory effect is crucial in determining the long time behavior of the rigid body motion; in fact, artificially neglecting the history part in $D(t)$ leads to exponential decay of $V(t)$~\cite[p. 171]{Caprino2006}.

\begin{figure}[htbp]
  \centering
  \begin{tikzpicture}[scale=0.63]
    \fill[black,opacity=0.23] (0,0) circle (0.5 and 2);
    \fill[gray!30!black,opacity=0.25] (0,2) -- (2,2) arc (90:270:0.5 and 2) -- (0,-2) arc (-90:90:0.5 and 2);
    \fill[gray!60!black,opacity=0.25] (2,0) circle (0.5 and 2);
    \draw[->,>=stealth,dotted] (-4,0) -- (6,0);
    \draw[->,>=stealth,thick] (2,0) -- (3,0) node[above] {$V(t)$};
  \end{tikzpicture}
  \caption{A cylindrical rigid body is moving in one direction with velocity $V(t)$.}
  \label{fig:cylinder}
\end{figure}
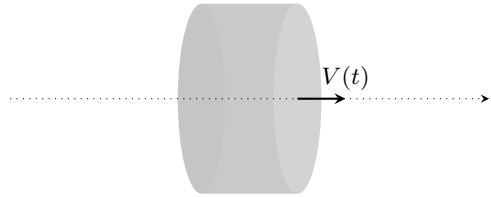

The memory effect in free molecular flow is caused by microscopic dynamics called \textit{recollision}: multiple collisions of a molecule with the body~(Fig.~\ref{fig:recollision}). Recollision causes the velocity distribution of molecules on the surface of the body at time $t$ dependent on the history of $V(t)$; therefore, the drag $D(t)$ also depends on the history of motion --- memory effect.

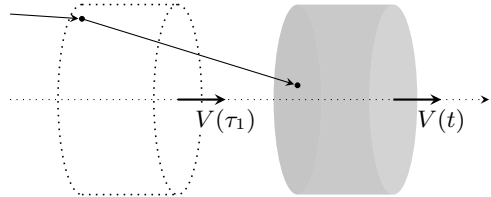
\begin{figure}[htbp]
  \centering
  \begin{tikzpicture}[scale=0.63]
    \draw[dotted,semithick] (-2.5,2) arc (90:270:0.5 and 2);
    \draw[dotted,semithick] (-2.5,2) -- (-0.5,2);
    \draw[dotted,semithick] (-2.5,-2) -- (-0.5,-2);
    \draw[dotted,semithick] (-0.5,2) arc (90:340:0.5 and 2);
    \draw[dotted,semithick] (-0.5,2) arc (90:-5:0.5 and 2);
    \draw[->,>=stealth,thick] (-0.5,0) -- (0.5,0) node[below] {$V(\tau_1)$};

    \fill[black,opacity=0.23] (2,0) circle (0.5 and 2);
    \fill[gray!30!black,opacity=0.25] (2,2) -- (4,2) arc (90:270:0.5 and 2) -- (2,-2) arc (-90:90:0.5 and 2);
    \fill[gray!60!black,opacity=0.25] (4,0) circle (0.5 and 2);
    \draw[->,>=stealth,dotted] (-4,0) -- (6,0);
    \draw[->,>=stealth,thick] (4,0) -- (5,0) node[below] {$V(t)$};

    \filldraw[fill=black,draw=black] (-2.5,1.7) circle (0.05);
    \filldraw[fill=black,draw=black] (2,0.3) circle (0.05);
    \draw[->,>=stealth,shorten >=1.3] (-4,1.8) -- (-2.5,1.7);
    \draw[->,>=stealth,shorten >=1.3] (-2.5,1.7) -- (2,0.3);
  \end{tikzpicture}
  \caption{A molecule (black dot) colliding with the body multiple times. Here $\tau_1<t$.}
  \label{fig:recollision}
\end{figure}

The central question of this paper is the following: What happens to the memory effect if the molecules have certain interaction with background obstacles?\footnote{If the obstacles are molecules themselves and the interaction is an elastic collision, then the resulting kinetic equation is the linearized Boltzmann equation~\cite{Gallavotti}.} Intuitively, the memory effect will be lost. Consider a molecule colliding with the body at time $t$ and $\tau_1(<t)$. This time, however, there may be interaction with the obstacles in-between this recollision (see Fig.~\ref{fig:scattering}). If there are sufficient scattering by the obstacles, then the velocities of the molecule at time $t-0$ and at $\tau_1+0$ are likely to be uncorrelated; therefore, the information of the history of motion ($V(\tau_1)$ in this situation) is not conveyed to time $t$. And the memory effect in $D(t)$ should be lost, which then results in a qualitative change in the long time behavior.

\begin{figure}[htbp]
  \centering
  \begin{tikzpicture}[scale=0.63]
    \draw[dotted,semithick] (-2.5,2) arc (90:270:0.5 and 2);
    \draw[dotted,semithick] (-2.5,2) -- (-0.5,2);
    \draw[dotted,semithick] (-2.5,-2) -- (-0.5,-2);
    \draw[dotted,semithick] (-0.5,2) arc (90:340:0.5 and 2);
    \draw[dotted,semithick] (-0.5,2) arc (90:-5:0.5 and 2);
    \draw[->,>=stealth,thick] (-0.5,0) -- (0.5,0) node[below] {$V(\tau_1)$};
    \fill[black,opacity=0.23] (2,0) circle (0.5 and 2);
    \fill[gray!30!black,opacity=0.25] (2,2) -- (4,2) arc (90:270:0.5 and 2) -- (2,-2) arc (-90:90:0.5 and 2);
    \fill[gray!60!black,opacity=0.25] (4,0) circle (0.5 and 2);
    \draw[dotted] (-4,0) -- (6,0);
    \draw[->,>=stealth,thick] (4,0) -- (5,0) node[below] {$V(t)$};
    \filldraw[fill=black,draw=black] (-2.5,1.7) circle (0.05);
    \filldraw[fill=black,draw=black] (2,0.3) circle (0.05);
    \draw[dotted] (-4,1.8) -- (-2.5,1.7);
    \draw[dotted] (-2.5,1.7) -- (1,0.63);
    \draw[->,>=stealth,shorten >=1.3] (1,0.63) -- (2,0.3);
    \filldraw[fill=black,draw=black] (1,0.63) circle (0.05);
    \draw[->,>=stealth,shorten >=1.3] (-2.5,0.15) -- (1,0.63);
    \filldraw[fill=black,draw=black] (-2.5,0.15) circle (0.05);
    \draw[->,>=stealth,shorten >=1.3] (-3.5,-1) -- (-2.5,0.15);
    \filldraw[fill=black,draw=black] (1,0.8) circle (0.05);
    \draw[->,>=stealth,shorten >=1.6] (0.3,1.5) -- (1,0.77);
    \draw[->,>=stealth,shorten >=1.3] (1,0.8) -- (1.5,1.2);
  \end{tikzpicture}
  \caption{A molecule may have interaction with the obstacles in-between a recollision.}
  \label{fig:scattering}
\end{figure}
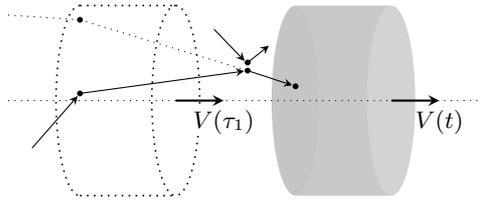

This question was raised and numerically analyzed by Tsuji and Aoki in~\cite{Tsuji2014} assuming rather special interaction with the background obstacles. The resulting kinetic equation is called a special \textit{Lorentz gas}.\footnote{See~\cite{Gallavotti} for a description of Lorentz gas in general.} They showed that the velocity $V(t)$ of the rigid body decays exponentially if the obstacles undergo thermal motion; if they are motionless, then $V(t)$ decays algebraically with a rate independent of the spatial dimension ($t^{-4}$ in their setting).

The purpose of this paper is to give a mathematical proof of this numerical observation. The setting, however, is slightly different from theirs. In this paper, the boundary condition for the kinetic equation is the specular reflection instead of the diffuse reflection and there is no linear restoring force applied to the rigid body; this is more close to the original setting in~\cite{Caprino2006,Caprino2007}. Nevertheless, qualitatively the same asymptotic behavior as in Theorem~\ref{thm:large_kappa} is proved. The proof uses the methods developed in~\cite{Caprino2006,Caprino2007} with additional decay estimates of the recollision terms with the help of a semi-explicit solution formula for the special Lorentz gas (Lemma~\ref{lem:fW_formula}).

The outline of the paper is as follows: I explain the formulation in the next section. The main theorem is stated in Section~\ref{sec:theorems} and is proved in Section~\ref{sec:proof_thm:kappa_large}. Some discussion of the problem is given in Section~\ref{sec:discussion}.

\section{Motion of a Rigid Body in a Lorentz Gas}
\label{sec:formulation}
This section gives the equations governing the motion of a rigid body and the surrounding gas.

Consider a rigid body in $\mathbb{R}^3$ whose section by the $d$-dimensional plane $\mathbb{R}^d$ ($d=1,2,3$) is
\begin{equation}
  \label{C(t)}
  C(t)=\{ x=(x_1,x_{\perp})\in \mathbb{R}\times \mathbb{R}^{d-1} \mid |x_1-X(t)|\leq h/2,\, |x_{\perp}|\leq 1 \},\footnote{Set $x=x_1$ and ignore the condition $|x_{\perp}|\leq 1$ when $d=1$.}
\end{equation}
where $h>0$ is a constant and $X(t)$ is a function of time ($X(0)=0$). That is, the rigid body is either a cylinder ($d=3$); a plate with infinite extension in the $x_3$-direction ($d=2$); or a plane wall with infinite extension both in the $x_2$ and $x_3$-directions ($d=1$). Denote by $V(t)=dX(t)/dt$ the velocity of the rigid body.

A gas fills the region outside the rigid body and is described by the velocity distribution function $f(x,\xi,t)$. Here, $x\in \Omega(t)\coloneqq \mathbb{R}^d \backslash C(t)$,
\begin{equation}
  \label{xi}
  \xi=(\xi_1,\xi_{\perp},\hat{\xi})=(\tilde{\xi},\hat{\xi})\in (\mathbb{R}\times \mathbb{R}^{d-1})\times \mathbb{R}^{3-d}
\end{equation}
and $t\geq 0$. Note that the velocity variable $\xi$ is three-dimensional even when $d=1$ or $2$. The gas is modeled as a special kind of Lorentz gas which is introduced in~\cite{Tsuji2012}. I explain this model below only briefly; the reader can find a more detailed account in their paper.

The gas consists of monatomic gas molecules and randomly dispersed obstacles. A crucial assumption of the model is that the distribution of the obstacles is not disturbed by the presence of the gas molecules and is given by a spatially homogeneous Maxwellian; only the velocity distribution function $f$ of the gas molecules changes in time. The evolution law of $f$ is determined by specifying the interaction of the gas molecules with the obstacles (the model assumes that the gas molecules are so dilute so that the interaction of the gas molecules with itself can be neglected). This model assumes that the obstacles behave like the condensed phase of the gas: The molecules hitting the obstacles are absorbed and re-emitted from them. And the velocity $\xi$ of an emitted molecule from an obstacle moving with velocity $\xi_s$ follows a Maxwellian distribution $f_0(\xi-\xi_s)$, where $f_0(\xi)=f(x,\xi,0)$ is the initial distribution of the gas molecules --- meaning that the gas molecules are initially saturated.

The kinetic equation for $f$ is derived in~\cite{Tsuji2012} by a standard argument in the kinetic theory of gases under the assumptions stated above and the requirement that
\begin{equation}
  \label{epsilon}
  \varepsilon=\frac{\text{average speed of the obstacles}}{\text{average speed of the molecules}}
\end{equation}
is small. Taking into account only terms up to $O(\varepsilon)$ and writing in suitable dimensionless variables,\footnote{In this paper, the length scale $\sqrt{2R_* T_{*0}}/\omega_*$ in~\cite{Tsuji2012} is replaced by the size of the rigid body, which is the radius of the cylinder when $d=3$; the width of the plate when $d=2$; and an arbitrary positive constant when $d=1$. All other scales are the same as in~\cite{Tsuji2012}. Note that eq.~\eqref{C(t)} is already written in these dimensionless variables so that the radius of the cylinder when $d=3$ (or the width of the plate when $d=2$) is set equal to unity.} the kinetic equation is:
\begin{equation}
  \label{Lorentz}
  \partial_t f+\tilde{\xi}\cdot \nabla_x f=\frac{\nu_{\varepsilon}(|\xi|)}{\kappa}[\pi^{-3/2}\exp(-|\xi|^2)-f]
\end{equation}
for $x\in \Omega(t)$, $\xi \in \mathbb{R}^3$ and $t>0$. Here, $\kappa>0$ is the Knudsen number for the collisions between the molecules and the obstacles. The function $\nu_{\varepsilon}(z)$ ($\varepsilon \geq 0$) is defined for $z>0$ by
\begin{equation}
  \label{nu_eps}
  \nu_{\varepsilon}(z)=\frac{\varepsilon}{2}\left[ \exp(-z^2/\varepsilon^2)+\pi^{1/2}\left( \frac{z}{\varepsilon}+\frac{\varepsilon}{2z} \right) \erf(z/\varepsilon) \right],
\end{equation}
where $\erf(z)=2\pi^{-1/2}\int_{0}^{z}\exp(-y^2)\, dy$ is the error function; the value at $z=0$ is defined by $\nu_{\varepsilon}(0)=2\varepsilon /\pi^{1/2}$. Note that the definition of $\nu_{\varepsilon}(z)$ differs from that in~\cite{Tsuji2012} by a constant, which is absorbed in $\kappa$. Note that eq.~\eqref{Lorentz} also depends on $\hat{\xi}$, which is why the velocity variable $\xi$ is three-dimensional even when $d=1$ or $2$ --- unlike the spacial variable $x\in \mathbb{R}^d$.

Although the expression for $\nu_{\varepsilon}(z)$ is rather complex, the only properties of $\nu_{\varepsilon}(z)$ used in this paper are its continuity and the bounds:
\begin{equation}
  \label{nu_bounds}
  \varepsilon+C^{-1}\frac{z^2}{\varepsilon+z}\leq \nu_{\varepsilon}(z)\leq \varepsilon+C\frac{z^2}{\varepsilon+z} \quad (z\geq 0)
\end{equation}
for some $C\geq 1$. Note that by eq.~\eqref{nu_eps}, $\nu_{\varepsilon}(z)\to z$ as $\varepsilon \to 0$; and $\nu_{\varepsilon}'(0)=0$ and $\nu_{\varepsilon}''(0)=2/(3\varepsilon)$; therefore, the limit $\varepsilon \to 0$ is singular in the sense that the first derivative of $\nu_{\varepsilon}(z)$ vanishes at $z=0$ for $\varepsilon>0$ but $\nu_{0}(z)=z$. And this is reflected in the bounds~\eqref{nu_bounds}.

The initial condition is:
\begin{equation}
  \label{init_gas}
  f(x,\xi,0)=f_0(\xi)\coloneqq \pi^{-3/2}\exp(-|\xi|^2)
\end{equation}
for $x\in \Omega(0)$ and $\xi \in \mathbb{R}^3$.

The rigid body motion affects the gas dynamics through a boundary condition for eq.~\eqref{Lorentz}: the specular boundary condition in this paper. Let $\be_1=(1,0,0)\in \mathbb{R}^3$ and $\bn=(\bn_d,\bm{0})\in \mathbb{R}^d \times \mathbb{R}^{3-d}$, where $\bn_d=\bn_d(x,t) \in \mathbb{R}^d$ is the unit normal to $\partial C(t)$ at $x\in \partial C(t)$ pointing towards the gas. Then the specular boundary condition is:
\begin{equation}
  \label{specular}
  f(x,\xi,t)=f(x,\xi-2[(\xi-V(t)\be_1)\cdot \bn]\bn,t)
\end{equation}
for $x\in \partial C(t)$, $\xi \in \mathbb{R}^3$ with $(\xi-V(t)\be_1)\cdot \bn>0$ and $t>0$.

On the other hand, the gas affects the rigid body motion through the gas drag: The velocity $V(t)$ of the rigid body is governed by Newton's equation
\begin{equation}
  \label{Newton}
  \frac{dV(t)}{dt}=-D(t),\quad V(0)=V_0,
\end{equation}
where $D(t)$ is the gas drag given by
\begin{equation}
  \label{drag}
  D(t)=\int_{\partial C(t)}\, dS\int \xi_1(\xi-V(t)\be_1)\cdot \bn f\, d\xi
\end{equation}
and $V_0>0$ is the initial velocity. Formula~\eqref{drag} is derived by considering the net momentum flux of the molecules at $\partial C(t)$ (see~\cite{Sone2007}).\footnote{The interaction of the obstacles with the rigid body is not considered as in~\cite{Tsuji2012}.} Let
\begin{equation}
  \label{I_pm}
  I^{\pm}(t)=\{ x\in C(t) \mid x_1=X(t)\pm h/2 \} \times \{ \xi \in \mathbb{R}^3 \mid \xi_1 \lessgtr V(t) \}.
\end{equation}
Using boundary condition~\eqref{specular}, $D(t)$ is also written as
\begin{equation}
  \label{drag_rewritten}
  D(t)=2\left( \int_{I^{+}(t)}(\xi_1-V(t))^2 f\, d\xi dS-\int_{I^{-}(t)}(\xi_1-V(t))^2 f\, d\xi dS \right).
\end{equation}
Note that the lateral side of $\partial C(t)$ does not contribute to $D(t)$.

Solving eqs.~\eqref{Lorentz}, \eqref{init_gas} and~\eqref{specular}; and eq.~\eqref{Newton} with eq.~\eqref{drag_rewritten} determines the motion of the rigid body and the surrounding gas. These equations are coupled in both ways: Boundary condition~\eqref{specular} requires the knowledge of the velocity $V(t)$ and computing the drag $D(t)$ requires the gas state $f$.

\section{Theorem on the Long Time Behavior}
\label{sec:theorems}
The problem I discuss in this paper is the long time behavior of $V(t)$, which is the content of Theorem~\ref{thm:large_kappa}. This gives a mathematical basis of the numerical observation given in~\cite{Tsuji2012}.

In order to state the theorem (and proving it), I use a function $D_0 \colon \mathbb{R}\to \mathbb{R}$ defined by
\begin{equation}
  \label{D0}
  D_0(U)=c_d \left( \int_{-\infty}^{U}(u-U)^2 e^{-u^2}\, du-\int_{U}^{\infty}(u-U)^2 e^{-u^2}\, du \right),
\end{equation}
where $c_1=2\pi^{-1/2}$, $c_2=4\pi^{-1/2}$ and $c_2=2\pi^{1/2}$. Note that
\begin{align}
  \label{D0_origin}
  \begin{aligned}
    & D_0(V(t)) \\
    & =2\left( \int_{I^{+}(t)}(\xi_1-V(t))^2 f_0 \, d\xi dS-\int_{I^{-}(t)}(\xi_1-V(t))^2 f_0 \, d\xi dS \right),
  \end{aligned}
\end{align}
which explains the choice of the constant $c_d$. The following lemma gives some properties of $D_0$; its proof is easy (see~\cite{Caprino2006}).

\begin{lemma}
  \label{lem:D0}
  $D_0 \colon \mathbb{R}\to \mathbb{R}$ is convex on the interval $[0,\infty)$; and it is odd, smooth and uniformly increasing on $\mathbb{R}$.
\end{lemma}

I use the following notation: $C_0=D_{0}'(0)$, $C_{\gamma}=D_{0}'(\gamma)$ ($0<\gamma \leq 1$), $t_{\gamma}=(\log \gamma^{-1})/C_{\gamma}$ and
\begin{equation}
  \label{w}
  w_{\varepsilon,\kappa,d}(t)=\frac{1}{(1+t)^{d+2}}\left( \frac{1}{\sqrt{1+t/(\varepsilon \kappa)}}+\frac{1}{1+t/\kappa} \right)^{3-d}.
\end{equation}
Note that $0<C_0<C_{\gamma} \leq D_{0}'(1)$ by Lemma~\ref{lem:D0}. Using the convention that $1/0=\infty$ and $1/\infty=0$,
\begin{equation}
  \label{w_epsilon=0}
  w_{0,\kappa,d}(t)=\frac{1}{(1+t)^{d+2}(1+t/\kappa)^{3-d}}.
\end{equation}
Moreover, in the limit of $\kappa \to \infty$,
\begin{equation}
  \label{w_kappa=infty}
  \lim_{\kappa \to \infty}w_{\varepsilon,\kappa,d}(t)=\frac{2^{3-d}}{(1+t)^{d+2}}.
\end{equation}

The following theorem is the main result of this paper, which gives the long time behavior of $V(t)$.

\begin{theorem}
  \label{thm:large_kappa}
  Suppose that $\kappa \geq 1$ and $\varepsilon \leq \kappa C_0/4$. Then for $\gamma>0$ sufficiently small, there exists a solution $(f,V)$ to eqs.~\eqref{Lorentz}, \eqref{init_gas} and~\eqref{specular}; and eq.~\eqref{Newton} with $V_0=\gamma$ satisfying the following inequalities:
  \begin{align} 
    V(t) & \geq \gamma e^{-C_{\gamma}t}-\gamma^3 A_1 w_{\varepsilon,\kappa,d}(t)e^{-\frac{\varepsilon}{2\kappa}t}, \label{lower_thm1} \\
    V(t) & \leq \gamma e^{-C_0 t}-\gamma^5 A_2 w_{\varepsilon,\kappa,d}(t)e^{-\frac{\varepsilon}{\kappa}t}\bm{1}_{\{ t\geq 2t_{\gamma} \}}, \label{upper_thm1}
  \end{align}
  where $A_1$ and $A_2$ are positive constants depending only on $d$. Moreover, any solution $(f,V)$ satisfies these inequalities and $V$ is decreasing on the interval $[0,t_{\gamma}]$.
\end{theorem}

\begin{remark}
  \label{rem:thm1}
  \begin{enumerate}[label=(\roman*)]
    \item If $\varepsilon \neq 0$, $V(t)$ decays exponentially; if $\varepsilon=0$, $V(t)$ decays algebraically with a rate $-5$ (i.e., $V(t)\approx t^{-5}$), which is independent of the spatial dimension $d$ (see also the discussion in Section~\ref{sec:discussion}); in the limit of $\kappa \to \infty$, $V(t)$ decays algebraically with a rate $-(d+2)$, which is exactly the result in the free molecular case~\cite{Caprino2007}.
    \item The uniqueness of the solution is unknown as in~\cite{Caprino2006,Caprino2007}; however, at least the uniqueness of the long time behavior is guaranteed by the theorem.
    \item $V(t)$ changes its sign: $V(t)>0$ for $t\leq t_{\gamma}$ and $V(t)<0$ for $t\geq 8t_{\gamma}$ (taking $\gamma$ sufficiently small if necessary). This is similar to the free molecular case~\cite{Caprino2007}.
    \item If $\varepsilon$ is too large, Theorem~\ref{thm:large_kappa} does not hold; in fact, if $\varepsilon \geq 2\kappa C_0$, then $V(t)$ is always positive and decays monotonically and exponentially. This theorem is stated in Section~\ref{sec:discussion} (Theorem~\ref{thm:small_kappa}) and proved in the appendix.
  \end{enumerate}
\end{remark}

\section{Proof of Theorem~\ref{thm:large_kappa}}
\label{sec:proof_thm:kappa_large}
First, I set the notation and explain the outline of the proof.

Let $W\colon [0,\infty)\to \mathbb{R}$ be an arbitrary Lipschitz continuous function and put $X_W(t)=\int_{0}^{t}W(s)\, ds$. Define $C_W(t)$ as the right hand side of eq.~\eqref{C(t)} but with $X(t)$ replaced by $X_W(t)$; similarly, define $I_{W}^{\pm}(t)$ as the right hand side of eq.~\eqref{I_pm} but with $C(t)$, $X(t)$ and $V(t)$ replaced by $C_W(t)$, $X_W(t)$ and $W(t)$.

Denote by $f=f_W$ the solution to eqs.~\eqref{Lorentz}, \eqref{init_gas} (with $\Omega(t)$ replaced by $\Omega_W(t)=\mathbb{R}^d \backslash C_W(t)$) and the specular boundary condition
\begin{equation}
  \label{specular_W}
  f(x,\xi,t)=f(x,\xi-2[(\xi-W(t)\be_1)\cdot \bn]\bn,t)
\end{equation}
for $x\in C_W(t)$, $\xi \in \mathbb{R}^3$ with $(\xi-W(t)\be_1 )\bn>0$ and $t>0$, where $\bn$ is now the unit normal to $\partial C_W(t)$. $f_W$ is constructed explicitly by the method of characteristics in Section~\ref{sec:characteristics}.

Using $f_W$, I shall define another function $V_W \colon [0,\infty)\to \mathbb{R}$ as follows: First, define $r_{W}^{\pm}(t)$ by
\begin{equation}
  \label{rpm}
  r_{W}^{\pm}(t)=\pm 2\int_{I_{W}^{\pm}(t)}(\xi_1-W(t))^2 (f_W-f_0)\, d\xi dS.
\end{equation}
Next, define $K\colon \mathbb{R}\to [0,\infty)$ by
\begin{equation}
  \label{K}
  K(U)=D_0(U)/U
\end{equation}
for $U\neq 0$; if $U=0$, define $K(0)=D_{0}'(0)$. Now, define $V_W \colon [0,\infty)\to \mathbb{R}$ by solving the equations
\begin{equation}
  \label{VW}
  \frac{dV_W(t)}{dt}=-K(W(t))V_W(t)-r_{W}^{+}(t)-r_{W}^{-}(t),\quad V_W(0)=\gamma.
\end{equation}
This is solved explicitly:
\begin{equation}
  \label{VW_formula}
  V_W(t)=\gamma e^{-\int_{0}^{t}K(W(s))\, ds}-\int_{0}^{t}e^{-\int_{s}^{t}K(W(\tau))\, d\tau}(r_{W}^{+}(s)+r_{W}^{-}(s))\, ds.
\end{equation}

Suppose that $V$ is a fixed point of the map $W\mapsto V_W$. Then $(f_V,V)$ solves eqs.~\eqref{Lorentz}, \eqref{init_gas} and~\eqref{specular}; and eq.~\eqref{Newton} with $V_0=\gamma$. This is verified easily using eq.~\eqref{D0_origin}.

To show the existence of a fixed point, the map $W\mapsto V_W$ must be analyzed in a suitable function space:

\begin{definition}
  \label{def:space1}
  Let $\gamma$, $A_1$ and $A_2$ be positive constants. A Lipschitz continuous function $W\colon [0,\infty)\to \mathbb{R}$ belongs to $\mathcal{K}=\mathcal{K}(\gamma,A_1,A_2)$ if $W(0)=\gamma$; $W$ is decreasing on the interval $[0,t_{\gamma}]$; and satisfies
  \begin{align}
    W(t) & \geq \gamma e^{-C_{\gamma}t}-\gamma^3 A_1 w_{\varepsilon,\kappa,d}(t)e^{-\frac{\varepsilon}{2\kappa}t}, \label{W_lower} \\
    W(t) & \leq \gamma e^{-C_0 t}-\gamma^5 A_2 w_{\varepsilon,\kappa,d}(t)e^{-\frac{\varepsilon}{\kappa}t}\bm{1}_{\{ t\geq 2t_{\gamma} \}} \label{W_upper}
  \end{align}
  and $|dW(t)/dt|\leq 1$.\footnote{The dependence of $\mathcal{K}$ on the parameters $\varepsilon$, $\kappa$ and $d$ is omitted for notational simplicity.}
\end{definition}

The plan of the proof is as follows: First, I prove appropriate decay estimates of $r_{W}^{\pm}(t)$ in Section~\ref{sec:rWpm_bounds} with the help of preparations in Sections~\ref{sec:characteristics}, \ref{sec:xi_estimates} and~\ref{sec:integral_estimates}. By using the decay estimates, I show in Section~\ref{sec:fixed_point} that $W\in \mathcal{K}$ implies $V_W \in \mathcal{K}$ if $\gamma$ is sufficiently small. Then Schauder's fixed point theorem is applied to show the existence of a fixed point $V\in \mathcal{K}$ (Section~\ref{sec:fixed_point}), and $(f_V,V)$ is a solution to the equations; since $V\in \mathcal{K}$, $V$ satisfies ineqs.~\eqref{lower_thm1} and \eqref{upper_thm1} by definition~\ref{def:space1}. Finally, ``any solution'' part of Theorem~\ref{thm:large_kappa} is proved in Section~\ref{sec:any_solution}.

\subsection{Analysis of $f_W$ by the Method of Characteristics}
\label{sec:characteristics}
The solution $f=f_W$ to eqs.~\eqref{Lorentz}, \eqref{init_gas} and~\eqref{specular_W} can be constructed by the method of characteristics as follows.

Let $W\colon [0,\infty)\to \mathbb{R}$ be an arbitrary Lipschitz continuous function and let $(x,\xi)\in I_{W}^{\pm}(t)$. The characteristics $(x(s),\xi(s))=(x(s,t;x,\xi),\xi(s,t;x,\xi))$ starting from $(x,\xi)$ are defined as follows: Let $x(s)=x-(t-s)\tilde{\xi}$ and $\xi(s)=\xi$ for $s \leq t$ until $x(s)$ hits the boundary $\partial C_W(s)$ --- denote this time of \textit{recollision} (or \textit{precollision}) as $\tau_1$; if such recollision do not occur at positive time, then put $\tau_1=0$. Thus $(x(s),\xi(s))$ is defined for $\tau_1 \leq s\leq t$.

If $\tau_1>0$, define the specularly reflected velocity $\xi'(\tau_1)\in \mathbb{R}^3$ by
\begin{equation}
  \label{reflected_velocity}
  \xi'(\tau_1)=(2W(\tau_1)-\xi_1,\xi_{\perp},\hat{\xi}).
\end{equation}
Then extend the characteristics as follows: Let $x(s)=x(\tau_1)-(\tau_1-s)\tilde{\xi}'(\tau_1)$ and $\xi(s)=\xi'(\tau_1)$ for $s<\tau_1$ until $x(s)$ again hits the boundary --- denote this time of recollision as $\tau_2$; if such recollision do not occur at positive time, then put $\tau_2=0$. Thus $(x(s),\xi(s))$ is defined for $\tau_2 \leq s<\tau_1$.

Repeat this to define $\tau_n$ until $\tau_{N+1}=0$ for some $N\geq 0$. Such $N$ exists for a.e. $(x,\xi)\in I_{W}^{\pm}(t)$ for each $t>0$ by~\cite[Proposition A.1]{Caprino2006}: Infinite or tangential (meaning that $\xi(\tau_n)=W(\tau_n)$) recollisions are measure theoretically negligible.

Using $\{ \tau_n \}$ constructed above, an explicit formula for $f_W$ by the following lemma:

\begin{lemma}
  \label{lem:fW_formula}
  If $\tau_1>0$, then
  \begin{align}
    \label{fW_formula}
    \begin{aligned}
      & f_W(x,\xi,t)-f_0(\xi) \\
      & =\sum_{n=1}^{N}(f_0(\xi'(\tau_n))-f_0(\xi(\tau_{n})))\prod_{i=1}^{n}\exp\left( -\frac{\nu_{\varepsilon}(|\xi(\tau_i)|)}{\kappa}(\tau_{i-1}-\tau_{i}) \right).
    \end{aligned}
  \end{align}
  for a.e. $(x,\xi)\in I_{W}^{\pm}(t)$; if $\tau_1=0$, then $f_W(x,\xi,t)=f_0(\xi)$.
\end{lemma}

\begin{proof}
  By eq.~\eqref{Lorentz} (with $\Omega(t)$ replaced by $\Omega_W(t)=\mathbb{R}^d \backslash C_W(t)$), it easily follows that
  \begin{align}
    \label{f-f0_tau1}
    \begin{aligned}
      & f_W(x,\xi,t)-f_0(\xi) \\
      & =(f_W(x(\tau_1),\xi(\tau_1),\tau_1)-f_0(\xi))\exp\left( -\frac{\nu_{\varepsilon}(|\xi|)}{\kappa}(t-\tau_1) \right).
    \end{aligned}
  \end{align}
  If $\tau_1=0$ (meaning that no recollision occurred), then the right hand side equals zero by eq.~\eqref{init_gas}; if $\tau_1>0$, then eq.~\eqref{specular_W} implies
  \begin{equation}
    \label{fW_reflection}
    f_W(x(\tau_1),\xi(\tau_1),\tau_1)=f_W(x(\tau_1),\xi'(\tau_1),\tau_1)
  \end{equation}
  since $\bn(x(\tau_1),\tau_1)=\pm \be_1$ in this case.
  
  Again, by eq.~\eqref{Lorentz},
  \begin{align}
    \label{f-f0_tau1_tau2}
    \begin{aligned}
      & f_W(x(\tau_1),\xi'(\tau_1),\tau_1)-f_0(\xi'(\tau_1)) \\
      & =(f_W(x(\tau_2),\xi(\tau_2),\tau_2)-f_0(\xi'(\tau_1)))\exp\left( -\frac{\nu_{\varepsilon}(|\xi'(\tau_1)|)}{\kappa}(\tau_1-\tau_2) \right).
    \end{aligned}
  \end{align}
  Inserting eqs.~\eqref{fW_reflection} and~\eqref{f-f0_tau1_tau2} into \eqref{f-f0_tau1} gives
  \begin{align}
    \label{f-f0_tau2}
    \begin{aligned}
      & f_W(x,\xi,t)-f_0(\xi) \\
      & =(f_0(\xi'(\tau_1))-f_0(\xi))\exp\left( -\frac{\nu_{\varepsilon}(|\xi|)}{\kappa}(t-\tau_1) \right) \\
      & \quad +(f_W(x(\tau_2),\xi(\tau_2),\tau_2)-f_0(\xi'(\tau_1))) \\
      & \quad \quad \times \exp\left( -\frac{\nu_{\varepsilon}(|\xi|)}{\kappa}(t-\tau_1) \right) \exp\left( -\frac{\nu_{\varepsilon}(|\xi'(\tau_1)|)}{\kappa}(\tau_1-\tau_2) \right).
    \end{aligned}
  \end{align}
  Repeating this argument and using eq.~\eqref{init_gas} gives eq.~\eqref{fW_formula}. \qed
\end{proof}

Equation~\eqref{fW_formula} can be used to prove the following bound of $f_W-f_0$.

\begin{lemma}
  \label{lem:fW_bound}
  For a.e. $(x,\xi)\in I_{W}^{\pm}(t)$,
  \begin{equation}
    \label{fW_bound}
    |f_W(x,\xi,t)-f_0(\xi)|\leq \pi^{-3/2}e^{-|(\xi_{\perp},\hat{\xi})|^2}\exp\left( -\frac{\nu_{\varepsilon}(|\xi|)}{\kappa}(t-\tau_1) \right) \bm{1}_{\{ \tau_1>0 \}}.
  \end{equation}
  Here, $\{ \tau_1>0 \}=\{ (x,\xi)\in I_{W}^{\pm}(t) \mid \tau_1>0 \}$.
\end{lemma}

\begin{proof}
  By Lemma~\ref{lem:fW_formula}, $\tau_1=0$ implies $f_W(x,\xi,t)=f_0(\xi)$ and ineq.~\eqref{fW_bound} trivially holds; so I assume $\tau_1>0$ in the following.

  Let $(x,\xi)\in I_{W}^{+}(t)$. (The case of $(x,\xi)\in I_{W}^{-}(t)$ is similar.) Note that by definition $W(t)>\xi_1$. And for a recollision to occur at time $\tau_1>0$, the rigid body must be moving faster than the molecule at time $\tau_1$: $\xi_1=\xi_1(\tau_1)>W(\tau_1)$; therefore, $W(\tau_1)>\xi_{1}'(\tau_1)=\xi_1(\tau_2)$ by eq.~\eqref{reflected_velocity}. Repeating this argument leads to
  \begin{equation}
    \label{xi1_W}
    W(\tau_{k-1})>\xi_1(\tau_k)>W(\tau_k)>\xi_{1}'(\tau_k)=\xi_1(\tau_{k+1})
  \end{equation}
  for $1\leq k\leq N$, where I used the convention that $\tau_0=t$.

  Suppose there exists $1\leq k\leq N$ such that $W(\tau_k)<0$. And let $k_0$ be the smallest such $k$; if such $k$ does not exist, let $k_0=N+1$. Then
  \begin{align}
    W(\tau_k) & \geq 0 \quad (1\leq k<k_0), \label{W_sign1} \\
    W(\tau_k) & <0 \quad (k_0\leq k\leq N) \label{W_sign2}
  \end{align}
  by ineqs.~\eqref{xi1_W}. Since
  \begin{equation}
    \label{xi_prime_squared}
    |\xi_{1}'(\tau_k)|^2=|\xi_1(\tau_k)|^2+4W(\tau_k)(W(\tau_k)-\xi_1(\tau_k)),
  \end{equation}
  ineqs.~\eqref{W_sign1} and~\eqref{W_sign2} imply
  \begin{align}
    |\xi_{1}'(\tau_k)| & \leq |\xi_1(\tau_k)| \quad (1\leq k<k_0), \label{xi_sign1} \\
    |\xi_{1}'(\tau_k)| & >|\xi_1(\tau_k)| \quad (k_0\leq k\leq N). \label{xi_sign2}
  \end{align}
  These are equivalent to
  \begin{align}
    f_0(\xi'(\tau_k)) & \geq f_0(\xi(\tau_k)) \quad (1\leq k<k_0), \label{f0_sign1} \\
    f_0(\xi'(\tau_k)) & <f_0(\xi(\tau_k)) \quad (k_0 \leq k\leq N); \label{f0_sign2}
  \end{align}
  therefore, by Lemma~\ref{lem:fW_formula},
  \begin{align}
    \label{f-f0_upper}
    \begin{aligned}
      & f_W(x,\xi,t)-f_0(\xi) \\
      & \leq \sum_{k=1}^{k_0-1}(f_0(\xi'(\tau_k))-f_0(\xi(\tau_k)))\exp\left( -\frac{\nu_{\varepsilon}(|\xi(\tau_1)|)}{\kappa}(t-\tau_1) \right) \\
      & =(f_0(\xi'(\tau_{k_0-1}))-f_0(\xi))\exp\left( -\frac{\nu_{\varepsilon}(|\xi|)}{\kappa}(t-\tau_1) \right) \\
      & \leq \pi^{-3/2}e^{-|(\xi_{\perp},\hat{\xi})|^2}\exp\left( -\frac{\nu_{\varepsilon}(|\xi|)}{\kappa}(t-\tau_1) \right).
    \end{aligned}
  \end{align}
  Similarly,
  \begin{align}
    \label{f-f0_lower}
    \begin{aligned}
      & f_W(x,\xi,t)-f_0(\xi) \\
      & \geq \sum_{k=k_0}^{N}(f_0(\xi'(\tau_k))-f_0(\xi(\tau_k)))\exp\left( -\frac{\nu_{\varepsilon}(|\xi(\tau_1)|)}{\kappa}(t-\tau_1) \right) \\
      & =(f_0(\xi'(\tau_{N}))-f_0(\xi(\tau_{k_0})))\exp\left( -\frac{\nu_{\varepsilon}(|\xi|)}{\kappa}(t-\tau_1) \right) \\
      & \geq \pi^{-3/2}e^{-|(\xi_{\perp},\hat{\xi})|^2}\exp\left( -\frac{\nu_{\varepsilon}(|\xi|)}{\kappa}(t-\tau_1) \right).
    \end{aligned}
  \end{align}
  These inequalities prove the lemma. \qed
\end{proof}

\subsection{Bounds of $\xi$ Leading to Recollisions}
\label{sec:xi_estimates}
Let $W\in \mathcal{K}=\mathcal{K}(\gamma,A_1,A_2)$ and $(x,\xi)\in I_{W}^{\pm}(t)$. In order for the characteristics $(x(s),\xi(s))$ starting from $(x,\xi)$ to have at least one recollision ($\tau_1>0$), $\xi$ must satisfy certain bounds, which I explain in the following.

First, note that if $\tau_1>0$, then
\begin{equation}
  \label{xi1_W_rel}
  (t-\tau_1)\xi_1=\int_{0}^{t}W(s)\, ds.
\end{equation}
Note that if I use a notation
\begin{equation}
  \label{Waver}
  \langle W \rangle_{s,t}=\frac{1}{t-s}\int_{s}^{t}W(\tau)\, d\tau \quad (0\leq s\leq t),\footnote{$\langle W \rangle_{t,t}=W(t)$.}
\end{equation}
then eq.~\eqref{xi1_W_rel} is equivalent to $\xi_1=\langle W \rangle_{\tau_1,t}$.

\begin{lemma}
  \label{lem:xi_estimates}
  Let $W\in \mathcal{K}$, $0<\eta<1$ and $(x,\xi)\in I_{W}^{\pm}(t)$. If $0<\tau_1 \leq \eta t$, then
  \begin{equation}
    \label{xi1_estimate1}
    -C\gamma^3 A_1 \frac{e^{-\frac{\varepsilon}{2\kappa}\tau_1}}{1+t}\leq \xi_1 \leq C\frac{\gamma}{1+t}
  \end{equation}
  and
  \begin{equation}
    \label{xiperp_estimate}
    |\xi_{\perp}|\leq \frac{C}{t},
  \end{equation}
  where $C$ is a positive constant depending only on $\eta$ and $d$; if $\eta t<\tau_1<t$, then
  \begin{equation}
    \label{xi1_estimate2}
    -C\gamma^3 A_1 w_{\varepsilon,\kappa,d}(t)e^{-\frac{\varepsilon}{2\kappa}\eta t}\leq \xi_1 \leq \gamma e^{-C_0 \eta t}.
  \end{equation}
\end{lemma}

\begin{proof}
  Applying ineq.~\eqref{W_lower} to eq.~\eqref{xi1_W_rel} gives
  \begin{align}
    \label{xi1_lower}
    \begin{aligned}
      \xi_1
      & \geq -\frac{\gamma^3 A_1}{t-\tau_1}\int_{\tau_1}^{t}w_{\varepsilon,\kappa,d}(s)e^{-\frac{\varepsilon}{2\kappa}s}\, ds.
    \end{aligned}
  \end{align}
  Similarly, ineq.~\eqref{W_upper} gives
  \begin{equation}
    \label{xi1_upper}
    \xi_1 \leq \frac{1}{t-\tau_1}\int_{\tau_1}^{t}\gamma e^{-C_0 s} \, ds.
  \end{equation}
  These show ineqs.~\eqref{xi1_estimate1} and \eqref{xi1_estimate2}. (Note that $w_{\varepsilon,\kappa,d}(\eta t)\leq Cw_{\varepsilon,\kappa,d}(t)$ for some positive constant $C$.)

  Moreover, $\xi_{\perp}$ must satisfy $|x_{\perp}-(t-\tau_1)\xi_{\perp}|\leq 1$ if $\tau_1>0$; therefore,
  \begin{equation}
    \label{xiperp_estimate_step1}
    |\xi_{\perp}|\leq \frac{2}{t-\tau_1}.
  \end{equation}
  This shows ineq.~\eqref{xiperp_estimate} if $0<\tau_1<\eta t$. \qed
\end{proof}

\begin{remark}
  \label{rem:xi_estimates}
  Retaining the term $e^{-\varepsilon \tau_1/(2\kappa)}$ in the lower bound of ineq.~\eqref{xi1_estimate1} plays an important role in proving the upper bound or $|r_{W}^{+}(t)$ (Section~\ref{sec:|rWp|_upper}).
\end{remark}

\subsection{Bounds of an Integral Involving $\nu_{\varepsilon}(z)$}
\label{sec:integral_estimates}
A certain integral involving $\nu_{\varepsilon}(z)$ appears in the estimates of $r_{W}^{\pm}(t)$. To give bounds of this integral is the purpose of this section.

First, I need a preliminary lemma.

\begin{lemma}
  \label{lem:integral_preliminary}
  Let $d,\sigma \in \{ 1,2 \}$. The integral
  \begin{equation}
    \label{Idsigma}
    I_{d,\sigma}(t)=\int_{\hat{\xi}\in \mathbb{R}^{3-d}}e^{-|\hat{\xi}|^2}e^{-|\hat{\xi}|^{\sigma}t}\, d\hat{\xi}
  \end{equation}
  satisfies
  \begin{equation}
    \label{Idsigma_estimates}
    \frac{C^{-1}}{(1+t)^{(3-d)/\sigma}}\leq I_{d,\sigma}(t)\leq \frac{C}{(1+t)^{(3-d)/\sigma}}
  \end{equation}
  for some positive constant $C$.
\end{lemma}

\begin{proof}
  First, let us consider the case of $d=1$:
  \begin{equation}
    \label{I1sigma}
    I_{1,\sigma}(t)=2\pi \int_{0}^{\infty}re^{-r^2}e^{-r^{\sigma}t}\, dr.
  \end{equation}
  If $\sigma=2$, then
  \begin{equation}
    \label{I12}
    I_{1,2}(t)=2\pi \int_{0}^{\infty}re^{-(1+t)r^2}\, dr=\frac{\pi}{1+t},
  \end{equation}
  which shows ineq.~\eqref{Idsigma_estimates}. If $\sigma=1$, then
  \begin{align}
    \label{I11}
    \begin{aligned}
      I_{1,1}(t)
      & =2\pi \int_{0}^{\infty}re^{-(r+t/2)^2}e^{t^2/4}\, dr \\
      & =2\pi e^{t^2/4}\int_{t/2}^{\infty}(\lambda-t/2)e^{-\lambda^2}\, d\lambda \\
      & =\pi \left( 1-te^{t^2/4}\int_{t/2}^{\infty}e^{-\lambda^2}\, d\lambda \right).
    \end{aligned}
  \end{align}
  Repeated use of the relation $e^{-\lambda^2}=-(2\lambda)^{-1}(e^{-\lambda^2})'$ shows that
  \begin{equation}
    \label{erfc_asymptotic_expansion}
    te^{t^2/4}\int_{t/2}^{\infty}e^{-\lambda^2}\, d\lambda=1-\frac{2}{t^2}+\frac{3}{4}te^{t^2/4}\int_{t/2}^{\infty}\frac{1}{\lambda^4}e^{-\lambda^2}\, d\lambda;
  \end{equation}
  therefore,
  \begin{equation}
    \label{I11_continued}
    I_{1,1}(t)=\frac{2\pi}{t^2}+O(t^{-3}) \quad \text{as $t\to \infty$},
  \end{equation}
  which shows ineqs.~\eqref{Idsigma_estimates}.

  The case of $d=2$ is reduced to that of $d=1$ by using
  \begin{equation}
    \label{I2sigma_squared}
    (I_{2,\sigma}(t))^2=2\pi \int_{\mathbb{R}^2}e^{-(x^2+y^2)}e^{-(|x|^{\sigma}+|y|^{\sigma})t}\, dxdy
  \end{equation}
  and $(|x|+|y|)/\sqrt{2}\leq \sqrt{x^2+y^2}\leq |x|+|y|$. \qed
\end{proof}

The following lemma is the main result of this section.

\begin{lemma}
  \label{lem:integral}
  Let
  \begin{equation}
    \label{J}
    J=J(\tilde{\xi},t)=\int_{\hat{\xi}\in \mathbb{R}^{3-d}}e^{-|\hat{\xi}|^2}e^{-\frac{\nu_{\varepsilon}(|\xi|)}{\kappa}t}\, d\hat{\xi}.
  \end{equation}
  Then
  \begin{equation}
    \label{integral_estimate1}
    J\leq C\left( \frac{1}{\sqrt{1+t/(\varepsilon \kappa)}}+\frac{1}{1+t/\kappa} \right)^{3-d} e^{-\frac{\varepsilon}{\kappa}t}
  \end{equation}
  and
  \begin{equation}
    \label{integral_estimate2}
    J\geq C^{-1}e^{-C'|\tilde{\xi}|\frac{t}{\kappa}}\left( \frac{1}{\sqrt{1+t/(\varepsilon \kappa)}}+\frac{1}{1+t/\kappa} \right)^{3-d} e^{-\frac{\varepsilon}{\kappa}t}
  \end{equation}
  for some constants $C,C'\geq 1$.
\end{lemma}

\begin{proof}
  First, I prove the upper bound~\eqref{integral_estimate1}. By the lower bound in ineqs.~\eqref{nu_bounds},
  \begin{align}
    \label{integral_estimates_step1}
    \begin{aligned}
      e^{-\frac{\nu_{\varepsilon}(|\xi|)}{\kappa}t}
      & \leq e^{-\frac{\varepsilon}{\kappa}t}e^{-C^{-1}\frac{|\xi|^2}{\varepsilon+|\xi|}\frac{t}{\kappa}} \\
      & \leq e^{-\frac{\varepsilon}{\kappa}t}\left( e^{-\frac{|\xi|}{2C}\frac{t}{\kappa}}+e^{-\frac{|\xi|^2}{2C\varepsilon}\frac{t}{\kappa}} \right) \\
      & \leq e^{-\frac{\varepsilon}{\kappa}t}\left( e^{-\frac{|\hat{\xi}|}{2C}\frac{t}{\kappa}}+e^{-\frac{|\hat{\xi}|^2}{2C\varepsilon}\frac{t}{\kappa}} \right).
    \end{aligned}
  \end{align}
  Now Lemma~\ref{lem:integral_preliminary} gives the upper bound~\eqref{integral_estimate1}.

  The lower bound~\eqref{integral_estimate2} is proved as follows. Let
  \begin{equation}
    \label{alpha}
    \alpha_{\varepsilon}(z)=\frac{z^2}{\varepsilon+z} \quad (z\geq 0).
  \end{equation}
  Then $\alpha_{\varepsilon}(z)$ is increasing in $z$ and satisfies
  \begin{align}
    \alpha_{\varepsilon}(z) & \leq z, \label{alpha_bound1} \\
    \alpha_{\varepsilon}(z) & \leq \frac{z}{2}+\frac{z^2}{2\varepsilon}, \label{alpha_bound2} \\
    \alpha_{\varepsilon}(z+w) & \leq 2(\alpha_{\varepsilon}(z)+\alpha_{\varepsilon}(w)) \label{alpha_almost_subadditive}
  \end{align}
  for $z,w\geq 0$. By using these properties of $\alpha_{\varepsilon}(z)$ and the upper bound in ineqs.~\eqref{nu_bounds},
  \begin{align}
    \label{integral_estimates_step2}
    \begin{aligned}
      e^{-\frac{\nu_{\varepsilon}(|\xi|)}{\kappa}t}
      & \geq e^{-\frac{\varepsilon}{\kappa}t}e^{-C\alpha_{\varepsilon}(|\xi|)\frac{t}{\kappa}} \\
      & \geq e^{-\frac{\varepsilon}{\kappa}t}e^{-C\alpha_{\varepsilon}(|\tilde{\xi}|+|\hat{\xi}|)\frac{t}{\kappa}} \\
      & \geq e^{-\frac{\varepsilon}{\kappa}t}e^{-2C\alpha_{\varepsilon}(|\tilde{\xi}|)\frac{t}{\kappa}}e^{-2C\alpha_{\varepsilon}(|\hat{\xi}|)\frac{t}{\kappa}} \\
      & \geq e^{-\frac{\varepsilon}{\kappa}t}e^{-2C|\tilde{\xi}|\frac{t}{\kappa}}\left( e^{-C|\hat{\xi}|\frac{t}{\kappa}}+e^{-C\frac{|\hat{\xi}|^2}{\varepsilon}\frac{t}{\kappa}} \right).
    \end{aligned}
  \end{align}
  Lemma~\ref{lem:integral_preliminary} gives the lower bound~\eqref{integral_estimate2}. \qed
\end{proof}

\subsection{Decay Estimates of $r_{W}^{\pm}(t)$}
\label{sec:rWpm_bounds}
\subsubsection{Upper Bound of $|r_{W}^{+}(t)|$}
\label{sec:|rWp|_upper}
In this section, I prove an upper bound of $|r_{W}^{+}(t)|$, which is given in the following proposition.

\begin{proposition}
  \label{prop:|rWp|_upper}
  Let $\kappa \geq 1$, $\varepsilon \leq \kappa C_0/4$ and $W\in \mathcal{K}(\gamma,A_1,A_2)$. If $\gamma$ is sufficiently small so that $\gamma \leq \min \{ 1,1/A_1 \}$ and $C_0/C_{\gamma}\geq 9/10$, then
  \begin{equation}
    \label{|rWp|_upper}
    |r_{W}^{+}(t)|\leq C\gamma^{\frac{21}{4}}w_{\varepsilon,\kappa,d}(t)e^{-\frac{\varepsilon}{\kappa}t}\bm{1}_{\{ t\geq t_{\gamma} \}}
  \end{equation}
  for some positive constant $C$ independent of $\gamma$, $A_1$ and $A_2$.
\end{proposition}

\begin{proof}
  Suppose that $t<t_{\gamma}$. Then since $W$ is decreasing on the interval $[0,t_{\gamma}]$, $\tau_1=0$ for any $(x,\xi)\in I_{W}^{+}(t)$; therefore, $f_W(x,\xi,t)=f_0(\xi)$ by Lemma~\ref{lem:fW_formula}, and hence $r_{W}^{+}(t)=0$ by eq.~\eqref{rpm}. So I assume in the rest of the proof that $t\geq t_{\gamma}$. In particular, by ineq.~\eqref{W_upper} and $C_0/C_{\gamma}\geq 9/10$,
  \begin{equation}
    \label{W_t_geq_t0}
    W(t)\leq \gamma e^{-\frac{5}{6}C_0 t_{\gamma}}e^{-\frac{C_0}{6}t}=\gamma^{1+\frac{5C_0}{6C_{\gamma}}t}e^{-\frac{C_0}{6}t}\leq \gamma^{\frac{7}{4}}e^{-\frac{C_0}{6}t}.
  \end{equation}

  Now let
  \begin{equation}
    \label{Rt}
    R_t=\{ (x,\xi)\in I_{W}^{+}(t) \mid \tau_1>0 \}.
  \end{equation}
  By definition, $\tau_1=0$ for any $(x,\xi)\in I_{W}^{+}(t)\backslash R_t$, which implies $f_W(x,\xi,t)=f_0(\xi)$ by Lemma~\ref{lem:fW_formula}; therefore,
  \begin{equation}
    \label{rW+}
    r_{W}^{+}(t)=2\int_{R_t}(\xi_1-W(t))^2 (f_W-f_0)\, d\xi dS.
  \end{equation}
  Next, divide $R_t$ into two parts: $R_t=R_{t}'\cup R_{t}''$, where
  \begin{equation}
    \label{Rt_decomposition}
    R_{t}'=\{ (x,\xi)\in R_t \mid \tau_1 \leq 2t/3 \}, \quad R_{t}''=\{ (x,\xi)\in R_t \mid \tau_1>2t/3 \};
  \end{equation}
  correspondingly, $r_{W}^{+}(t)=r_{W,1}^{+}(t)+r_{W,2}^{+}(t)$, where
  \begin{align}
    r_{W,1}^{+}(t) & =2\int_{R_{t}'}(\xi_1-W(t))^2 (f_W-f_0)\, d\xi dS, \label{rW1p} \\
    r_{W,2}^{+}(t) & =2\int_{R_{t}''}(\xi_1-W(t))^2 (f_W-f_0)\, d\xi dS. \label{rW2p}
  \end{align}

  First, I prove the following:
  \begin{equation}
    \label{|rW1p|_upper}
    |r_{W,1}^{+}(t)|\leq C\gamma^{\frac{21}{4}}w_{\varepsilon,\kappa,d}(t)e^{-\frac{\varepsilon}{\kappa}t}.
  \end{equation}
  Let $(x,\xi)\in R_{t}'$ ($\tau_1 \leq 2t/3$). By Lemma~\ref{lem:xi_estimates} (the lower bound in ineqs.~\eqref{xi1_estimate1}), ineq.~\eqref{W_t_geq_t0}, $\varepsilon/\kappa \leq C_0/4$ and $\gamma A_1 \leq 1$,
  \begin{equation}
    \label{W-xi1_Rt'}
    0<W(t)-\xi_1 \leq C\left( \gamma^{\frac{7}{4}}e^{-\frac{C_0}{6}t}+\gamma^3 A_1 \frac{e^{-\frac{\varepsilon}{2\kappa}\tau_1}}{1+t} \right) \leq C\gamma^{\frac{7}{4}}\frac{e^{-\frac{\varepsilon}{2\kappa}\tau_1}}{1+t}.
  \end{equation}
  This and Lemma~\ref{lem:xi_estimates} (ineq.~\eqref{xiperp_estimate}) give
  \begin{align}
    \label{Rt'_inclusion}
    \begin{aligned}
      R_{t}'
      & \subset \{ x\in C_{W}(t) \mid x_1=X_{W}(t)+h/2 \} \\ 
      & \quad \times \{ |\xi_1-W(t)|\leq C\gamma^{7/4}/(1+t) \} \times \{ |\xi_{\perp}|\leq C/t \} \times \mathbb{R}_{\hat{\xi}}^{3-d}.
    \end{aligned}
  \end{align}
  Moreover, ineq.~\eqref{W-xi1_Rt'} and Lemma~\ref{lem:fW_bound} give
  \begin{equation}
    \label{rW1+_integrand}
    (\xi_1-W(t))^2 |f_W-f_0|\leq C\gamma^{\frac{7}{2}}\frac{e^{-\frac{\varepsilon}{\kappa}\tau_1}}{(1+t)^2}e^{-|(\xi_{\perp},\hat{\xi})|^2}e^{-\frac{\nu_{\varepsilon}(|\xi|)}{\kappa}(t-\tau_1)}.
  \end{equation}
  Combining this with Lemma~\ref{lem:integral} and ineq.~\eqref{rW1+_integrand} show (note that $\tau_1$ is independent of $\hat{\xi}$)
  \begin{align}
    \label{rW1+_integrand_xihat}
    \begin{aligned}
      & \int_{\hat{\xi}}(\xi_1-W(t))^2 |f_W-f_0|\, d\hat{\xi} \\
      & \leq C\gamma^{\frac{7}{2}}e^{-|\xi_{\perp}|^2}\frac{e^{-\frac{\varepsilon}{\kappa}\tau_1}}{(1+t)^2}\left( \frac{1}{\sqrt{1+t/(3\varepsilon \kappa)}}+\frac{1}{1+t/(3\kappa)} \right)^{3-d}e^{-\frac{\varepsilon}{\kappa}(t-\tau_1)} \\
      & \leq C\gamma^{\frac{7}{2}}\frac{e^{-|\xi_{\perp}|^2}}{(1+t)^2}\left( \frac{1}{\sqrt{1+t/(\varepsilon \kappa)}}+\frac{1}{1+t/\kappa} \right)^{3-d}e^{-\frac{\varepsilon}{\kappa}t}.
    \end{aligned}
  \end{align}
  Combining this and inclusion~\eqref{Rt'_inclusion} proves ineq.~\eqref{|rW1p|_upper}.
  
  Next, I prove the following:
  \begin{equation}
    \label{|rW2p|_upper}
    |r_{W,2}^{+}(t)|\leq C\gamma^{\frac{21}{4}}(w_{\varepsilon,\kappa,d}(t))^3 e^{-\frac{\varepsilon}{\kappa}t}.
  \end{equation}
  Let $(x,\xi)\in R_{t}''$ ($\tau_1>2t/3$). By Lemma~\ref{lem:xi_estimates} (the lower bound in ineqs.~\eqref{xi1_estimate2}), ineq.~\eqref{W_t_geq_t0}, $\varepsilon/\kappa \leq C_0/4$ and $\gamma A_1 \leq 1$,
  \begin{align}
    \label{W-xi1_Rt''}
    \begin{aligned}
      0<W(t)-\xi_1
      & \leq C\left( \gamma^{\frac{7}{4}}e^{-\frac{C_0}{6}t}+\gamma^3 A_1 w_{\varepsilon,\kappa,d}(t)e^{-\frac{\varepsilon}{3\kappa}t} \right) \\
      & \leq C\gamma^{\frac{7}{4}}w_{\varepsilon,\kappa,d}(t)e^{-\frac{\varepsilon}{3\kappa}t}.
    \end{aligned}
  \end{align}
  This implies in particular the inclusion
  \begin{align}
    \label{Rt''_inclusion}
    \begin{aligned}
      R_{t}''
      & \subset \{ x\in C_{W}(t) \mid x_1=X_{W}(t)+h/2 \} \\ 
      & \quad \times \{ |\xi_1-W(t)|\leq C\gamma^{7/4}w_{\varepsilon,\kappa,d}(t)e^{-\varepsilon t/(3\kappa)} \} \times \mathbb{R}^2.
    \end{aligned}
  \end{align}
  This and ineq.~\eqref{W-xi1_Rt''} together with $|f_W-f_0|\leq \pi^{-3/2}e^{-|(\xi_{\perp},\hat{\xi})|^2}$ (which follows from Lemma~\ref{lem:fW_bound}) prove ineq.~\eqref{|rW2p|_upper}.

  Combining ineqs.~\eqref{|rW1p|_upper} and~\eqref{|rW2p|_upper} proves the lemma. \qed
\end{proof}

\subsubsection{Upper Bound of $|r_{W}^{-}(t)|$}
\label{sec:|rWm|_upper}
In this section, I prove an upper bound of $|r_{W}^{-}(t)|$:

\begin{proposition}
  \label{prop:|rWm|_upper}
  Let $\kappa \geq 1$, $\varepsilon \leq \kappa C_0/4$ and $W\in \mathcal{K}(\gamma,A_1,A_2)$. If $\gamma$ is sufficiently small so that $\gamma \leq \min \{ 1,1/A_1 \}$, then
  \begin{equation}
    \label{|rWm|_upper}
    |r_{W}^{-}(t)|\leq C\gamma^3 w_{\varepsilon,\kappa,d}(t)e^{-\frac{5\varepsilon}{6\kappa}t}
  \end{equation}
  for some positive constant $C$ independent of $\gamma$, $A_1$ and $A_2$.
\end{proposition}

\begin{proof}
  Let
  \begin{equation}
    \label{Lt}
    L_t=\{ (x,\xi)\in I_{W}^{-}(t) \mid \tau_1>0 \}
  \end{equation}
  and
  \begin{equation}
    \label{Lt_decomposition}
    L_{t}'=\{ (x,\xi)\in L_t \mid \tau_1 \leq t/6 \}, \quad L_{t}''=\{ (x,\xi)\in L_t \mid \tau_1>t/6 \}.
  \end{equation}
  Then as in the proof of Proposition~\ref{prop:|rWp|_upper}, $r_{W}^{-}(t)=r_{W,1}^{-}(t)+r_{W,2}^{-}(t)$, where
  \begin{align}
    r_{W,1}^{-}(t) & =2\int_{L_{t}'}(\xi_1-W(t))^2 (f_0-f_W)\, d\xi dS, \label{rW1m} \\
    r_{W,2}^{-}(t) & =2\int_{L_{t}''}(\xi_1-W(t))^2 (f_0-f_W)\, d\xi dS. \label{rW2m}
  \end{align}

  First, I prove the following:
  \begin{equation}
    \label{|rW1m|_upper}
    |r_{W,1}^{-}(t)|\leq C\gamma^3 w_{\varepsilon,\kappa,d}(t)e^{-\frac{5\varepsilon}{6\kappa}t}.
  \end{equation}
  Let $(x,\xi)\in L_{t}'$ ($\tau_1 \leq t/6$). By ineq.~\eqref{W_lower}, Lemma~\ref{lem:xi_estimates} (the upper bound in ineqs.~\eqref{xi1_estimate1}) and $\gamma A_1 \leq 1$,
  \begin{equation}
    \label{W-xi1_Lt'}
    0<\xi_1-W(t) \leq C\frac{\gamma}{1+t}+\gamma^3 A_1 w_{\varepsilon,\kappa,d}(t)e^{-\frac{\varepsilon}{2\kappa}t}\leq C\frac{\gamma}{1+t}.
  \end{equation}
  This and Lemma~\ref{lem:xi_estimates} (ineq.~\eqref{xiperp_estimate}) give
  \begin{align}
    \label{Lt'_inclusion}
    \begin{aligned}
      L_{t}'
      & \subset \{ x\in C_{W}(t) \mid x_1=X_{W}(t)-h/2 \} \\ 
      & \quad \times \{ |\xi_1-W(t)|\leq C\gamma/(1+t) \} \times \{ |\xi_{\perp}|\leq C/t \} \times \mathbb{R}_{\hat{\xi}}^{3-d}.
    \end{aligned}
  \end{align}
  Inequality~\eqref{W-xi1_Lt'} and Lemma~\ref{lem:fW_bound} give
  \begin{equation}
    \label{rW1-_integrand}
    (\xi_1-W(t))^2 |f_W-f_0|\leq C\frac{\gamma^2}{(1+t)^2}e^{-|(\xi_{\perp},\hat{\xi})|^2}e^{-\frac{5\nu_{\varepsilon}(|\xi|)}{6\kappa}t}.
  \end{equation}
  Now ineq.~\eqref{|rW1m|_upper} is proved by integrating ineq.~\eqref{rW1-_integrand} over $L_{t}'$ and using Lemma~\ref{lem:integral} and inclusion~\eqref{Lt'_inclusion}.
  
  Next, I prove the following:
  \begin{equation}
    \label{|rW2m|_upper}
    |r_{W,2}^{-}(t)|\leq C\gamma^3 (w_{\varepsilon,\kappa,d}(t))^3 e^{-\frac{3\varepsilon}{2\kappa}t}.
  \end{equation}
  Let $(x,\xi)\in L_{t}''$ ($\tau_1>t/6$). By ineq.~\eqref{W_lower}, Lemma~\ref{lem:xi_estimates} (the upper bound in ineqs.~\eqref{xi1_estimate2}), $\varepsilon/\kappa \leq C_0/4$ and $\gamma A_1 \leq 1$,
  \begin{equation}
    \label{W-xi1_Lt''}
    0<\xi_1-W(t)\leq \gamma e^{-\frac{C_0}{6}t}+\gamma^3 A_1 w_{\varepsilon,\kappa,d}(t)e^{-\frac{\varepsilon}{2\kappa}t}\leq C\gamma w_{\varepsilon,\kappa,d}(t)e^{-\frac{\varepsilon}{2\kappa}t}.
  \end{equation}
  This implies
  \begin{align}
    \label{Lt''_inclusion}
    \begin{aligned}
      L_{t}''
      & \subset \{ x\in C_{W}(t) \mid x_1=X_{W}(t)-h/2 \} \\ 
      & \quad \times \{ |\xi_1-W(t)|\leq C\gamma w_{\varepsilon,\kappa,d}(t)e^{-\varepsilon t/(2\kappa)} \} \times \mathbb{R}^2.
    \end{aligned}
  \end{align}
  This and ineq.~\eqref{W-xi1_Lt''} together with $|f_W-f_0|\leq \pi^{-3/2}e^{-|(\xi_{\perp},\hat{\xi})|^2}$ prove ineq.~\eqref{|rW2m|_upper}.
  
  Combining ineqs.~\eqref{|rW1m|_upper} and~\eqref{|rW2m|_upper} proves the lemma. \qed
\end{proof}

\subsubsection{Lower Bound of $r_{W}^{-}(t)$}
\label{sec:rWm_lower}
This section proves a lower bound of $r_{W}^{-}(t)$. First, I show a lemma.

\begin{lemma}
  \label{lem:s0}
  For $t>0$, define $s_0=s_0(t)$ by
  \begin{equation}
    \label{s0}
    s_0=\min \left\{ s\in (0,t) \mid W(s)\leq \frac{\gamma+\langle W \rangle_{s,t}}{2} \right\}.
  \end{equation}
  Then if $\gamma$ is sufficiently small and $t$ is sufficiently large (independent of $\gamma$),
  \begin{equation}
    \label{s0_bounds}
    \frac{1}{C_{\gamma}}\log \frac{6}{5}\leq s_0 \leq \frac{1}{C_0}\log \frac{8}{3}.
  \end{equation}
\end{lemma}

\begin{proof}
  Note first that $W(t)<\gamma$ ($t>0$) by ineq.~\eqref{W_upper}; therefore, the minimum in eq.~\eqref{s0} exists and $s_0$ is well-defined.

  First, I show that $s_0 \leq(\log 8/3)/C_0$. Note that by ineq.~\eqref{W_lower},
  \begin{align}
    \label{s0_step1}
    \begin{aligned}
      \langle W \rangle_{s_0,t}
      & \geq \frac{1}{t-s_0}\int_{s_0}^{t}\left( \gamma e^{-C_{\gamma}s}-\gamma^3 A_1 w_{\varepsilon,\kappa,d}(s)e^{-\frac{\varepsilon}{2\kappa}s} \right) \, ds \\
      & \geq -2^{3-d}\gamma^3 A_1.
    \end{aligned}
  \end{align}
  Take $\gamma$ sufficiently small so that $2^{3-d}\gamma^2 A_1 \leq 1/4$. Then by ineq.~\eqref{W_upper}, eq.~\eqref{s0} and ineq.~\eqref{s0_step1},
  \begin{equation}
    \label{s0_step2}
    \frac{3}{8}\gamma \leq \frac{\gamma+\langle W \rangle_{s_0,t}}{2}=W(s_0)\leq \gamma e^{-C_0 s_0};
  \end{equation}
  therefore, $s_0 \leq (\log 8/3)/C_0$.

  Next, I show that $s_0 \geq (\log 6/5)/C_{\gamma}$. By ineq.~\eqref{W_upper},
  \begin{equation}
    \label{s0_step3}
    \langle W \rangle_{s_0,t}\leq \frac{1}{t-s_0}\int_{s_0}^{t}\gamma e^{-C_0 s}\, ds=\gamma e^{-C_0 s_0}\frac{1-e^{-C_0(t-s_0)}}{C_0(t-s_0)}.
  \end{equation}
  Since $s_0 \leq (\log 8/3)/C_0$,
  \begin{equation}
    \label{s0_step4}
    \langle W \rangle_{s_0,t}\leq \frac{\gamma}{6}
  \end{equation}
  for sufficiently large $t$. On the other hand, by ineq.~\eqref{W_lower},
  \begin{align}
    \label{s0_step5}
    \begin{aligned}
      W(s_0)
      & \geq \gamma e^{-C_{\gamma}s_0}-\gamma^3 A_1 w_{\varepsilon,\kappa,d}(s_0)e^{-\frac{\varepsilon}{2\kappa}s_0} \\
      & \geq \gamma e^{-C_{\gamma}s_0}-2^{3-d}\gamma^3 A_1 \\
      & \geq \gamma e^{-C_{\gamma}s_0}-\frac{\gamma}{4};
    \end{aligned}
  \end{align}
  therefore,
  \begin{equation}
    \label{s0_step6}
    \gamma e^{-C_{\gamma}s_0}-\frac{\gamma}{4}\leq W(s_0)=\frac{\gamma+\langle W \rangle_{s_0,t}}{2}\leq \frac{7}{12}\gamma.
  \end{equation}
  This implies $s_0 \geq (\log 6/5)/C_{\gamma}$. \qed
\end{proof}

Next, let
\begin{equation}
  \label{St}
  S_t=\left\{ (x,\xi)\in I_{W}^{-}(t) \mid |x_{\perp}|\leq \frac{1}{2},\, \langle W \rangle_{s_0,t}<\xi_1<\langle W \rangle_{0,t},\, |\xi_{\perp}|\leq \frac{1}{2t} \right\}.
\end{equation}
Then $(x,\xi)\in S_t$ has exactly one recollision:

\begin{lemma}
  \label{lem:tau2=0}
  Let $(x,\xi)\in S_t$. If $\gamma$ is sufficiently small and $t$ is sufficiently large (independent of $\gamma$), then $0<\tau_1<s_0$ and $\tau_2=0$.
\end{lemma}

\begin{proof}
  Let $(x,\xi)\in S_t$. First, I show that $\tau_1>0$. Let $\sigma_1 \in (0,t)$ be the largest $s\in (0,t)$ such that $\xi_1=\langle W \rangle_{s,t}$ --- such $s$ exists since $\langle W \rangle_{s_0,t}<\xi_1<\langle W \rangle_{0,t}$ by definition of $S_t$. Then since
  \begin{equation}
    \label{tau2=0_step1}
    |x_{\perp}-(t-\sigma_1)\xi_{\perp}|\leq |x_{\perp}|+t|\xi_{\perp}|\leq 1
  \end{equation}
  again by definition of $S_t$, $\sigma_1$ coincides with $\tau_1$: $\sigma_1=\tau_1$. And $\tau_1>0$ since $\sigma_1$ is so.

  Next, I show that $\tau_1<s_0$. Take $\gamma$ sufficiently small so that $(\log 8/3)/C_0 \leq t_{\gamma}$. Then $W$ is decreasing on $[s_0,t_{\gamma}]$ by definition~\ref{def:space1} and Lemma~\ref{lem:s0}; therefore,
  \begin{equation}
    \label{tau2=0_step2}
    \langle W \rangle_{s,t}\leq \langle W \rangle_{s_0,t}<\xi_1=\langle W \rangle_{\tau_1,t}
  \end{equation}
  for $s\in [s_0,t_{\gamma}]$. This implies (by contradiction) that (i) $\tau_1<s_0$, or (ii) $\tau_1>t_{\gamma}$. In order to prove that (ii) is impossible (which proves $\tau_1<s_0$), it suffices to prove that
  \begin{equation}
    \label{tau2=0_step3}
    \langle W \rangle_{s,t}<\langle W \rangle_{s_0,t}
  \end{equation}
  for $s>t_{\gamma}$ since $\langle W \rangle_{\tau_1,t}=\xi_1>\langle W \rangle_{s_0,t}$ by definition of $S_t$. Ineq.~\eqref{tau2=0_step3} is prove as follows: By ineq.~\eqref{W_upper},
  \begin{equation}
    \label{tau2=0_step4}
    \langle W \rangle_{s,t}\leq \frac{1}{t-s}\int_{s}^{t}\gamma e^{-C_0 \tau}\, d\tau=\gamma e^{-C_0 s}\frac{1-e^{-C_0(t-s)}}{C_0(t-s)};
  \end{equation}
  on the other hand, by ineq.~\eqref{W_lower},
  \begin{align}
    \label{tau2=0_step5}
    \begin{aligned}
      \langle W \rangle_{s_0,t}
      & \geq \frac{1}{t-s_0}\int_{s_0}^{t}\left( \gamma e^{-C_{\gamma}s}-\gamma^3 A_1 w_{\varepsilon,\kappa,d}(s)e^{-\frac{\varepsilon}{2\kappa}s} \right) \, ds \\
      & \geq \gamma e^{-C_{\gamma}s_0}\frac{1-e^{-C_{\gamma}(t-s_0)}}{C_{\gamma}(t-s_0)}-C\frac{\gamma^3 A_1}{t-s_0} \\
      & \geq \frac{\gamma}{2}e^{-C_{\gamma}s_0}\frac{1-e^{-C_{\gamma}(t-s_0)}}{C_{\gamma}(t-s_0)}
    \end{aligned}
  \end{align}
  if $\gamma$ is sufficiently small and $t$ is sufficiently large. Now let $s>t_{\gamma}$ and consider first the case of $s\leq t/2$. Then by ineq.~\eqref{tau2=0_step4} and~\eqref{tau2=0_step5},
  \begin{equation}
    \label{tau2=0_step6}
    \langle W \rangle_{s,t}\leq \gamma e^{-C_0 t_{\gamma}}\frac{1-e^{-C_0 t/2}}{C_0 t/2}\leq \frac{\gamma}{2}e^{-C_{\gamma}s_0}\frac{1-e^{-C_{\gamma}(t-s_0)}}{C_{\gamma}(t-s_0)}\leq \langle W \rangle_{s_0,t}
  \end{equation}
  if $\gamma$ is sufficiently small and $t$ is sufficiently large, which proves ineq.~\eqref{tau2=0_step3} in this case; consider next the case of $s>t/2$. Then
  \begin{equation}
    \label{tau2=0_step7}
    \langle W \rangle_{s,t}\leq \gamma e^{-C_0 t/2}\leq \frac{\gamma}{2}e^{-C_{\gamma}s_0}\frac{1-e^{-C_{\gamma}(t-s_0)}}{C_{\gamma}(t-s_0)}\leq \langle W \rangle_{s_0,t}
  \end{equation}
  if $\gamma$ is sufficiently small and $t$ is sufficiently large, which proves ineq.~\eqref{tau2=0_step3} in this case.

  Now that $\tau_1<s_0$ (the case (i)) is proved, I next show that $\tau_2=0$: By eq.~\eqref{s0} and $0<\tau_1<s_0$,
  \begin{equation}
    \label{tau2=0_step8}
    \xi_{1}'(\tau_1)=2W(\tau_1)-\xi_1=2W(\tau_1)-\langle W \rangle_{\tau_1,t}>\gamma.
  \end{equation}
  Since $W(t)<\gamma$ ($t>0$) by ineq.~\eqref{W_upper}, more than one recollision is impossible: $\tau_2=0$. \qed
\end{proof}

A lower bound of $r_{W}^{-}(t)$ is given by the following proposition:

\begin{proposition}
  \label{prop:rWm_lower}
  Suppose that $\kappa \geq 1$, $\varepsilon \leq \kappa C_0/4$ and $W\in \mathcal{K}(\gamma,A_1,A_2)$. If $\gamma$ is sufficiently small, then
  \begin{equation}
    \label{rWm_lower}
    r_{W}^{-}(t) \geq C\gamma^5 w_{\varepsilon,\kappa,d}(t)e^{-\frac{\varepsilon}{\kappa}t}\bm{1}_{\{ t\geq t_{\gamma} \}}
  \end{equation}
  for some positive constant $C$ independent of $\gamma$, $A_1$ and $A_2$.
\end{proposition}

\begin{proof}
  First, note that $W(t)>0$ for $t\leq t_{\gamma}$ (see Remark~\ref{rem:thm1}~(iii)). This implies that
  \begin{equation}
    \label{rWm_lower_step0}
    |\xi_{1}'(\tau_k)|\geq |\xi_1(\tau_k)| \quad (1\leq k<N+1)
  \end{equation}
  for $(x,\xi)\in I_{W}^{-}(t)$ (see eq.~\eqref{xi_prime_squared}). This implies that $f_0(\xi)-f_W(x,\xi,t)\geq 0$ by Lemma~\ref{lem:fW_formula}; therefore, $r_{W}^{-}(t)\geq 0$ by eq.~\eqref{rpm}. This consideration legitimates the multiplication by $\bm{1}_{\{ t\geq t_{\gamma}\}}$ in ineq.~\eqref{rWm_lower}. Hence I can safely assume that $t\geq t_{\gamma}$ in the following. Moreover, in the following, I take $\gamma$ sufficiently small and $t$ sufficiently large so that Lemma~\ref{lem:tau2=0} holds: $0<\tau_1<s_0$ and $\tau_2=0$.

  Let
  \begin{align}
    \label{rWm_lower_step1}
    \begin{aligned}
      r_{W,3}^{-}(t) & =2\int_{S_t}(\xi_1-W(t))^2(f_0-f_W)\, d\xi dS, \\
      r_{W,4}^{-}(t) & =2\int_{I_{W}^{-}(t)\backslash S_t}(\xi_1-W(t))^2(f_0-f_W)\, d\xi dS.
    \end{aligned}
  \end{align}
  Then $r_{W}^{-}(t)=r_{W,3}^{-}(t)+r_{W,4}^{-}(t)$.

  First, I prove the following:
  \begin{equation}
    \label{rW3m_lower}
    r_{W,3}^{-}(t)\geq C\gamma^5 w_{\varepsilon,\kappa,d}(t)e^{-\frac{\varepsilon}{\kappa}t}\bm{1}_{\{ t\geq t_{\gamma} \}}.
  \end{equation}
  Let $(x,\xi)\in S_t$. Then
  \begin{equation}
    \label{rWm_lower_step2}
    C\gamma^2 \leq (\xi_{1}'(\tau_1))^2 -(\xi_1)^2 \leq 8 \gamma^2:
  \end{equation}
  Note first that
  \begin{equation}
    \label{rWm_lower_step3}
    (\xi_{1}'(\tau_1))^2 -(\xi_1)^2=(2W(\tau_1)-\xi_1)^2 -(\xi_1)^2=4W(\tau_1)(W(\tau_1)-\langle W \rangle_{\tau_1,t}).
  \end{equation}
  The upper bound in ineqs.~\eqref{rWm_lower_step2} follows from this since $W(t)<\gamma$ ($t>0$) by ineq.~\eqref{W_upper}. Next, since $\tau_1<s_0$,
  \begin{equation}
    \label{rWm_lower_step4}
    (\xi_{1}'(\tau_1))^2 -(\xi_1)^2 \geq 2W(\tau_1)(\gamma-\langle W \rangle_{\tau_1,t})
  \end{equation}
  by eq.~\eqref{s0} and~\eqref{rWm_lower_step3}. Take $\gamma$ sufficiently small so that $s_0 \leq t_{\gamma}$. Then since $W$ is decreasing on the interval $[0,t_{\gamma}]$,
  \begin{equation}
    \label{rWm_lower_step5}
    W(\tau_1)\geq W( (\log 8/3)/C_0) \geq C\gamma
  \end{equation}
  by ineq.~\eqref{W_lower}. Moreover, by ineq.~\eqref{W_upper},
  \begin{align}
    \label{rWm_lower_step6}
    \begin{aligned}
      \gamma-\langle W \rangle_{\tau_1,t}
      & =\gamma-\frac{1}{t-\tau_1}\int_{\tau_1}^{t}W(s)\, ds \\
      & \geq \gamma-\frac{1}{t-\tau_1}\int_{\tau_1}^{t}\gamma e^{-C_0 s}\, ds \\
      & =\gamma-\gamma e^{-C_0 \tau_1}\frac{1-e^{-C_0(t-\tau_1)}}{C_0(t-\tau_1)} \\
      & \geq C\gamma
    \end{aligned}
  \end{align}
  if $t$ is sufficiently large. Ineqs.~\eqref{rWm_lower_step4}, \eqref{rWm_lower_step5} and~\eqref{rWm_lower_step6} show the lower bound in ineqs.~\eqref{rWm_lower_step2}.

  Let us continue the proof of ineq.~\eqref{rW3m_lower}: Since $\tau_1>0$ and $\tau_2=0$,
  \begin{align}
    \label{rWm_lower_step7}
    \begin{aligned}
      f_0(\xi)-f_W(x,\xi,t)
      & =(f_0(\xi)-f_0(\xi'(\tau_1)))\exp\left( -\frac{\nu_{\varepsilon}(|\xi|)}{\kappa}(t-\tau_1) \right) \\
      & \geq C\gamma^2 e^{-(\xi_1)^2}e^{-|(\xi_{\perp},\hat{\xi})|^2}\exp\left( -\frac{\nu_{\varepsilon}(|\xi|)}{\kappa}t \right)
    \end{aligned}
  \end{align}
  by Lemma~\ref{lem:fW_formula} and ineq.~\eqref{rWm_lower_step2}; therefore, by Lemma~\ref{lem:integral} and eq.~\eqref{St},
  \begin{align}
    \label{rWm_lower_step8}
    \begin{aligned}
      r_{W,3}^{-}(t)
      & \geq C\gamma^2 \int_{\langle W \rangle_{s_0,t}}^{\langle W \rangle_{0,t}}e^{-(\xi_1)^2}(\xi_1-W(t))^2 \, d\xi_1 \\
      & \quad \times \int_{|\xi_{\perp}|\leq 1/(2t)}e^{-|\xi_{\perp}|^2}\, d\xi_{\perp}\int_{\hat{\xi}}e^{-|\hat{\xi}|^2}e^{-\frac{\nu_{\varepsilon}(|\xi|)}{\kappa}t}\, d\hat{\xi} \\
      & \geq C\gamma^2 \int_{\langle W \rangle_{s_0,t}}^{\langle W \rangle_{0,t}}e^{-(\xi_1)^2-C'|\xi_1|t}(\xi_1-W(t))^2 \, d\xi_1 \\
      & \quad \times \int_{|\xi_{\perp}|\leq 1/(2t)}e^{-|\xi_{\perp}|^2-C'|\xi_{\perp}|t}\, d\xi_{\perp} \\
      & \quad \times \left( \frac{1}{\sqrt{1+t/(\varepsilon \kappa)}}+\frac{1}{1+t/\kappa} \right)^{3-d}e^{-\frac{\varepsilon}{\kappa}t}.
    \end{aligned}
  \end{align}
  Note that
  \begin{equation}
    \label{rWm_lower_step9}
    0<\langle W \rangle_{s_0,t}<\langle W \rangle_{0,t}\leq \frac{\gamma}{C_0 t};
  \end{equation}
  therefore,
  \begin{align}
    \label{rWm_lower_step10}
    \begin{aligned}
      r_{W,3}^{-}(t)
      & \geq C\frac{\gamma^2}{(1+t)^{d-1}}\left( \frac{1}{\sqrt{1+t/(\varepsilon \kappa)}}+\frac{1}{1+t/\kappa} \right)^{3-d}e^{-\frac{\varepsilon}{\kappa}t} \\
      & \quad \times \int_{\langle W \rangle_{s_0,t}}^{\langle W \rangle_{0,t}}(\xi_1-W(t))^2 \, d\xi_1.
    \end{aligned}
  \end{align}
  The last integral is evaluated as
  \begin{equation}
    \label{rWm_lower_step11}
    \int_{\langle W \rangle_{s_0,t}}^{\langle W \rangle_{0,t}}(\xi_1-W(t))^2 \, d\xi_1=3^{-1}\{ (\langle W \rangle_{0,t}-W(t))^3-(\langle W \rangle_{s_0,t}-W(t))^3 \}.
  \end{equation}
  Note that
  \begin{align}
    \label{rWm_lower_step12}
    \begin{aligned}
      \langle W \rangle_{0,t}-\langle W \rangle_{s_0,t}
      & =\frac{1}{t-s_0}\int_{0}^{s_0}W(s)\, ds+\left( \frac{1}{t}-\frac{1}{t-s_0} \right)\int_{0}^{t}W(s)\, ds \\
      & =\frac{s_0}{t-s_0}(\langle W \rangle_{0,s_0}-\langle W \rangle_{0,t})
    \end{aligned}
  \end{align}
  and
  \begin{align}
    \label{rWm_lower_step13}
    \begin{aligned}
      \langle W \rangle_{0,s_0}-\langle W \rangle_{0,t}
      & \geq W(s_0)-\frac{1}{t}\int_{0}^{t}\gamma e^{-C_0 s}\, ds \\
      & \geq \gamma e^{-C_{\gamma}s_0}-2^{3-d}\gamma^3 A_1-\gamma \frac{1-e^{-C_0 t}}{C_0 t} \\
      & \geq C\gamma
    \end{aligned}
  \end{align}
  if $\gamma$ is sufficiently small and $t$ is sufficiently large; therefore,
  \begin{equation}
    \label{rWm_lower_step14}
    \langle W \rangle_{0,t}-\langle W \rangle_{s_0,t}\geq C\frac{\gamma}{t}.
  \end{equation}
  Moreover,
  \begin{align}
    \label{rWm_lower_step15}
    \begin{aligned}
      & \langle W \rangle_{s_0,t}-W(t) \\
      & \geq \frac{1}{t-s_0}\int_{s_0}^{t}\left( \gamma e^{-C_0 s}-\gamma^3 A_1 w_{\varepsilon,\kappa,d}(s)e^{-\frac{\varepsilon}{2\kappa}s} \right) \, ds-\gamma e^{-C_{\gamma}t} \\
      & \geq \gamma e^{-C_0 s_0}\frac{1-e^{-C_0 (t-s_0)}}{C_0 (t-s_0)}-\frac{2^{3-d}\gamma^3 A_1}{(d+1)(t-s_0)}-\gamma e^{-C_{\gamma}t} \\
      & \geq C\frac{\gamma}{t}
    \end{aligned}
  \end{align}
  if $\gamma$ is sufficiently large and $t$ is sufficiently large. Combining this with by eq.~\eqref{rWm_lower_step11} gives
  \begin{align}
    \label{rWm_lower_step16}
    \begin{aligned}
      \int_{\langle W \rangle_{s_0,t}}^{\langle W \rangle_{0,t}}(\xi_1-W(t))^2 \, d\xi_1
      & \geq (\langle W \rangle_{0,t}-\langle W \rangle_{s_0,t})(\langle W \rangle_{s_0,t}-W(t))^2 \\
      & \geq C\frac{\gamma^3}{(1+t)^3}.
    \end{aligned}
  \end{align}
  Now ineqs.~\eqref{rWm_lower_step10} and~\eqref{rWm_lower_step16} prove ineq.~\eqref{rW3m_lower}.

  Next, I prove the following:
  \begin{equation}
    \label{rW4m_lower}
    r_{W,4}^{-}(t)\geq -C\gamma^9 A_{1}^{3}(w_{\varepsilon,\kappa,d}(t))^3 e^{-\frac{3\varepsilon}{2\kappa}t}\bm{1}_{\{ t\geq t_{\gamma}\}}.
  \end{equation}
  Let $(x,\xi)\in S_{t}'\coloneqq I_{W}^{-}(t)\backslash S_t$. If $\xi_1>0$, then
  \begin{equation}
    \label{rWm_lower_step17}
    0<\xi_1<W(\tau_1)<\xi_{1}'(\tau_1)<\cdots<\xi_{1}'(\tau_{N-1})<W(\tau_N)<\xi_{1}'(\tau_N)
  \end{equation}
  and hence $f_0(\xi)-f_W(x,\xi,t)>0$ by Lemma~\ref{lem:fW_formula}; therefore,
  \begin{equation}
    \label{rWm_lower_step18}
    \int_{S_{t}'\cap \{ \xi_1>0 \}}(\xi_1-W(t))^2(f_0-f_W)\, d\xi dS \geq 0.
  \end{equation}
  So to prove ineq.~\eqref{rW4m_lower}, it suffices to prove that
  \begin{align}
    \label{rWm_lower_step19}
    \begin{aligned}
      & \int_{S_{t}'\cap \{ \xi_1<0 \}}(\xi_1-W(t))^2(f_0-f_W)\, d\xi dS \\
      & \geq -C\gamma^9 A_{1}^{3}(w_{\varepsilon,\kappa,d}(t))^3 e^{-\frac{3\varepsilon}{2\kappa}t}\bm{1}_{\{ t\geq t_{\gamma}\}}.
    \end{aligned}
  \end{align}
  Let $(x,\xi)\in S_{t}'\cap \{ \xi_1<0 \}$. Then by ineq.~\eqref{W_lower},
  \begin{align}
    \label{rWm_lower_step20}
    \begin{aligned}
      0
      & >W(t)-\xi_1 \\
      & \geq W(t) \\
      & \geq -\gamma^3 A_1 w_{\varepsilon,\kappa,d}(t)e^{-\frac{\varepsilon}{2\kappa}t}.
    \end{aligned}
  \end{align}
  This in particular implies the inclusion
  \begin{align}
    \label{St'_inclusion}
    \begin{aligned}
      S_{t}'\cap \{ \xi_1<0 \}
      & \subset \{ x\in C_{W}(t) \mid x_1=X_{W}(t)-h/2 \} \\ 
      & \quad \times \{ |\xi_1-W(t)|\leq \gamma^3 A_1 w_{\varepsilon,\kappa,d}(t)e^{-\varepsilon t/(2\kappa)} \} \times \mathbb{R}^2.
    \end{aligned}
  \end{align}
  Ineqs.~\eqref{St'_inclusion}, \eqref{rWm_lower_step20} and $|f_0-f_W|\leq \pi^{-3/2}e^{-|(\xi_{\perp},\hat{\xi})|^2}$ show ineq.~\eqref{rWm_lower_step19}, which implies ineq.~\eqref{rW4m_lower}.

  Combining ineqs.~\eqref{rW3m_lower} and~\eqref{rW4m_lower} proves the proposition. \qed
\end{proof}

\subsection{Existence of a Fixed Point}
\label{sec:fixed_point}
Applying the decay estimates of $r_{W}^{\pm}(t)$ obtained in Section~\ref{sec:rWpm_bounds}, I show in this section that the map $W\mapsto V_W$ (defined by eq.~\eqref{VW}) has a fixed point. The following proposition is the key to the proof.

\begin{proposition}
  \label{prop:VWinK}
  Suppose that $\kappa \geq 1$ and $\varepsilon \leq \kappa C_0/4$. Then there exists positive constants $A_1$ and $A_2$ independent of $\gamma$ such that $W\in \mathcal{K}=\mathcal{K}(\gamma,A_1,A_2)$ implies $V_W \in \mathcal{K}$ for sufficiently small $\gamma$.
\end{proposition}

\begin{proof}
  Let $W\in \mathcal{K}(\gamma,A_1,A_2)$, where $A_1$ and $A_2$ are specified later. Note that $|W(t)|\leq \gamma$ by ineqs.~\eqref{W_lower} and~\eqref{W_upper} (take $\gamma$ sufficiently small so that $2^{3-d}\gamma^2 A_1 \leq 1$); therefore, $C_0 \leq K(W(t))\leq C_{\gamma}$ by Lemma~\ref{lem:D0}.

  First, I show that $V_W$ satisfies ineq.~\eqref{W_lower} with $W(t)$ replaced by $V_W(t)$ strictly.\footnote{The strictness is needed later in Section~\ref{sec:any_solution}.} Put $R_W(t)=r_{W}^{+}(t)+r_{W}^{-}(t)$. Then by eq.~\eqref{VW_formula}, Propositions~\ref{prop:|rWp|_upper} and~\ref{prop:|rWm|_upper},
  \begin{align}
    \label{VWinK_step1}
    \begin{aligned}
      V_W(t)
      & \geq \gamma e^{-C_{\gamma}t}-\int_{0}^{\frac{3}{4}t}e^{-C_0(t-s)}R_W(s)\, ds-\int_{\frac{3}{4}t}^{t}e^{-C_0(t-s)}R_W(s)\, ds \\
      & \geq \gamma e^{-C_{\gamma}t}-C\gamma^3 e^{-\frac{C_0}{4}t}-C\gamma^3 w_{\varepsilon,\kappa,d}(t)e^{-\frac{5\varepsilon}{8\kappa}t} \\
      & \geq \gamma e^{-C_{\gamma}t}-\bar{C}\gamma^3 w_{\varepsilon,\kappa,d}(t)e^{-\frac{\varepsilon}{2\kappa}t},
    \end{aligned}
  \end{align}
  where $\bar{C}$ is a positive constant independent of $\gamma$ and $A_1$. Now put $A_1=2\bar{C}$. Then $V_W$ satisfies ineq.~\eqref{W_lower} strictly.

  Next, I show that $V_W$ satisfies ineq.~\eqref{W_upper} with $W(t)$ replaced by $V_W$ (strictly for $t\geq 0$). First, by Propositions~\ref{prop:|rWp|_upper} and~\ref{prop:rWm_lower},
  \begin{equation}
    \label{RW_lower}
    R_W(t)\geq C\gamma^5 w_{\varepsilon,\kappa,d}(t)e^{-\frac{\varepsilon}{\kappa}t}
  \end{equation}
  if $\gamma$ is sufficiently small. In particular, $R_W(t)\geq 0$; therefore, $V_W(t)<\gamma e^{-C_0 t}$ ($t>0$) follows from eq.~\eqref{VW_formula}. So I assume in the following that $t\geq 2t_{\gamma}$. By eq.~\eqref{VW_formula} and ineq.~\eqref{RW_lower},
  \begin{align}
    \label{VWinK_step2}
    \begin{aligned}
      V_W(t)
      & <\gamma e^{-C_0 t}-C\gamma^5 w_{\varepsilon,\kappa,d}(t)e^{-\frac{\varepsilon}{\kappa}t}\int_{t_{\gamma}}^{t}e^{-C_{\gamma}(t-s)}\, ds \\
      & =\gamma e^{-C_0 t}-C\gamma^5 w_{\varepsilon,\kappa,d}(t)e^{-\frac{\varepsilon}{\kappa}t}\frac{1-e^{-C_{\gamma}(t-t_{\gamma})}}{C_{\gamma}} \\
      & \leq \gamma e^{-C_0 t}-\tilde{C}\gamma^5 w_{\varepsilon,\kappa,d}(t)e^{-\frac{\varepsilon}{\kappa}t}
    \end{aligned}
  \end{align}
  if $\gamma$ is sufficiently small (so that $t_{\gamma}$ is sufficiently large), where $\tilde{C}$ is a positive constant independent of $\gamma$ and $A_2$. Put $A_2=\tilde{C}$. Then $V_W$ satisfies ineq.~\eqref{W_upper} (strictly for $t>0$).

  Next, I show that $V_W$ is decreasing on the interval $[0,t_{\gamma}]$. By differentiating eq.~\eqref{VW_formula},
  \begin{align}
    \label{VWinK_step3}
    \begin{aligned}
      \frac{dV_W(t)}{dt}
      & =-K(W(t))\gamma e^{-\int_{0}^{t}K(W(s))\, ds} \\
      & \quad +K(W(t))\int_{0}^{t}e^{-\int_{s}^{t}K(W(\tau))\, d\tau}R_W(s)\, ds-R_W(t).
    \end{aligned}
  \end{align}
  And by Propositions~\ref{prop:|rWp|_upper}, \ref{prop:|rWm|_upper} and $R_W(t)\geq 0$,
  \begin{equation}
    \label{VWinK_step4}
    \frac{dV_W(t)}{dt}\leq -C_0 \gamma e^{-C_{\gamma}t}+C\gamma^3<-\frac{C_0}{2}\gamma e^{-C_{\gamma}t}<0
  \end{equation}
  for $t\leq t_{\gamma}=(\log \gamma^{-1})/C_{\gamma}$ if $\gamma$ is sufficiently small.

  Finally, by eq.~\eqref{VW}, Propositions~\ref{prop:|rWp|_upper} and \ref{prop:|rWm|_upper},
  \begin{equation}
    \label{VWinK_step5}
    \left| \frac{dV_W(t)}{dt} \right| \leq C_{\gamma}\gamma+C\gamma^3 \leq 1
  \end{equation}
  if $\gamma$ is sufficiently small. \qed
\end{proof}

In what follows, I set $A_1=2\bar{C}$ and $A_2=\tilde{C}$, where $\bar{C}$ and $\tilde{C}$ are those appearing in the proof above.

The proposition just proved shows that the map
\begin{equation}
  \label{map_well_defined}
  \mathcal{K}\ni W\mapsto V_W \in \mathcal{K}
\end{equation}
is well-defined if $\gamma$ is sufficiently small. As claimed in the next proposition, this map is continuous in $C_b([0,\infty))$ --- the space of bounded continuous functions on the interval $[0,\infty)$ --- whose topology is defined by the norm $||W||=\sup_{0\leq t<\infty}|W(t)|$.

\begin{proposition}
  \label{prop:continuity}
  If $\{ W_j \}_{j=1}^{\infty}\subset \mathcal{K}$, $W\in \mathcal{K}$ and $W_k \to W$ in $C_b([0,\infty))$, then $V_{W_j}\to V_W$ in $C_b([0,\infty))$.
\end{proposition}

\begin{proof}
  Similar argument as in~\cite[pp. 179--180]{Caprino2006} shows that $R_{W_j}(t)\to R_W(t)$ as $j\to \infty$ for all $t\geq 0$. And by eq.~\eqref{VW_formula},
  \begin{align}
    \label{continuity_step1}
    \begin{aligned}
      & V_W(t)-V_{W_j}(t) \\
      & =\gamma\left( e^{-\int_{0}^{t}K(W(s))\, ds}-e^{-\int_{0}^{t}K(W_j(s))\, ds} \right) \\
      & \quad -\int_{0}^{t}e^{-\int_{s}^{t}K(W(\tau))\, d\tau}(R_W(s)-R_{W_j}(s))\, ds \\
      & \quad -\int_{0}^{t}\left( e^{-\int_{s}^{t}K(W(\tau))\, d\tau}-e^{-\int_{s}^{t}K(W_j(\tau))\, d\tau} \right) R_{W_j}(s)\, ds.
    \end{aligned}
  \end{align}
  I prove in the following that
  \begin{equation}
    \label{continuity_step2}
    \int_{0}^{t}e^{-\int_{s}^{t}K(W(\tau))\, d\tau}(R_W(s)-R_{W_j}(s))\, ds\to 0
  \end{equation}
  as $j\to \infty$ uniformly in $t\geq 0$. (The remaining first and third terms also vanishes as $j\to \infty$ uniformly in $t\geq 0$; I leave them to the reader.) By Propositions~\ref{prop:|rWp|_upper} and \ref{prop:|rWm|_upper}, $R_W(t)$ and $R_{W_j}(t)$ decay as $t\to \infty$ uniformly in $j$; therefore, for any $\delta>0$, there exists $T>0$ such that
  \begin{equation}
    \label{continuity_step3}
    \int_{T}^{t}e^{-C_0(t-s)}|R_W(s)-R_{W_j}(s)|\, ds<\delta/2
  \end{equation}
  for all $t\geq T$. By the Lebesgue dominated convergence theorem; the pointwise convergence $R_{W_j}(t)\to R_W(t)$ as $j\to \infty$; and the uniform boundedness of $R_W$ and $R_{W_j}$, there exists $N=N(T)\in \mathbb{N}$ such that
  \begin{equation}
    \label{continuity_step4}
  \int_{0}^{T}|R_W(s)-R_{W_j}(s)|\, ds<\delta/2
  \end{equation}
  for all $j\geq N$. Ineqs.~\eqref{continuity_step3} and~\eqref{continuity_step4} show that
  \begin{equation}
    \label{continuity_step5}
    \left| \int_{0}^{t}e^{-\int_{s}^{t}K(W(\tau))\, d\tau}(R_W(s)-R_{W_j}(s))\, ds \right|<\delta
  \end{equation}
  for $j\geq N$ and $t\geq 0$. This shows eq.~\eqref{continuity_step2} and proves the proposition. \qed
\end{proof}

\begin{remark}
  \label{rem:C1}
  In the proof above, I used the fact that $R_{W_j}(t)\to R_W(t)$ as $j\to \infty$; similar proof shows that $R_W(s)\to R_W(t)$ as $s\to t$: $R_W(t)$ is continuous in $t$. This implies by eq.~\eqref{VW} that $V_W \in C^1([0,\infty))$ for $W\in \mathcal{K}$.
\end{remark}

Now, note that $\mathcal{K}$ is a closed convex subset of $C_b([0,\infty))$. And since $|dW(t)/dt|\leq 1$ for all $W\in \mathcal{K}$, $\mathcal{K}$ is equi-continuous. Moreover, by ineqs.~\eqref{W_lower} and~\eqref{W_upper}, $\mathcal{K}$ is uniformly decaying as $t\to \infty$; therefore, by the Arzel\`{a}--Ascoli theorem (on a non-compact space $[0,\infty)$), $\mathcal{K}$ is compact in $C_b([0,\infty))$.

By the convexity, compactness of $\mathcal{K}$ and the continuity of the map $\mathcal{K}\in W\mapsto V_W \in \mathcal{K}$, Schauder's fixed theorem shows the existence of a fixed point $V\in \mathcal{K}$: $V=V_V$. Then $(f_V,V)$ is a solution to eqs.~\eqref{Lorentz}, \eqref{init_gas} and~\eqref{specular}; and eq.~\eqref{Newton} with $V_0=\gamma$. Since $V\in \mathcal{K}$, ineqs.~\eqref{lower_thm1} and~\eqref{upper_thm1} are satisfied.

\subsection{``Any Solution'' Part of Theorem~\ref{thm:large_kappa}}
\label{sec:any_solution}
Finally, this section proves that any solution $(f,V)$ satisfies ineqs.~\eqref{lower_thm1}, \eqref{upper_thm1} and $V$ is decreasing on the interval $[0,t_{\gamma}]$.

Let $(f,V)$ be a solution and let $\mathcal{F}$ be the set of $t>0$ such that
\begin{equation}
  \label{any_solution_step1}
  V(t)<\gamma e^{-C_{\gamma}t}-\gamma^3 A_1 w_{\varepsilon,\kappa,d}(t)e^{-\frac{\varepsilon}{2\kappa}t}
\end{equation}
or
\begin{equation}
  \label{any_solution_step2}
  V(t)>\gamma e^{-C_0 t}-\gamma^5 A_2 w_{\varepsilon,\kappa,d}(t)e^{-\frac{\varepsilon}{\kappa}t}\bm{1}_{\{ t\geq 2t_{\gamma} \}}.
\end{equation}
And let $T=\inf \mathcal{F}$.\footnote{The infimum of the empty set is $+\infty$.} 

I show first that $T>0$. Note that ineq.~\eqref{any_solution_step1} is not satisfied for some positive time since $V(0)=\gamma$ and $V\in C([0,\infty))$. Moreover, ineq.~\eqref{any_solution_step2} is also violated for some positive time $t$ since $\lim_{t\to +0}dV(t)/dt=-D_0(\gamma)<-C_0 \gamma$ and $V\in C^1([0,\infty))$ (see Remark~\ref{rem:C1}); therefore, $T>0$.

Next, let $\mathcal{F}'$ be the set of $t>0$ such that
\begin{equation}
  \label{any_solution_step3}
  \frac{dV(t)}{dt}\geq -\frac{C_0}{2}\gamma^2.
\end{equation}
And let $T'=\inf \mathcal{F}'$. Then $T'>0$ if $\gamma$ is sufficiently small: Ineq.~\eqref{any_solution_step3} is violated for some positive time since $\lim_{t\to +0}dV(t)/dt=-D_0(\gamma)$ and $V\in C^1([0,\infty))$. Note that
\begin{equation}
  \label{any_solution_step4}
  \frac{dV(T)}{dt}=-\frac{C_0}{2}\gamma^2
\end{equation}
and $V$ is decreasing on the interval $[0,T']$.

I show now that $T=+\infty$. So, suppose that $T<+\infty$. Then
\begin{equation}
  \label{any_solution_step5}
  V(t)\geq \gamma e^{-C_{\gamma}t}-\gamma^3 A_1 w_{\varepsilon,\kappa,d}(t)e^{-\frac{\varepsilon}{2\kappa}t}
\end{equation}
and
\begin{equation}
  \label{any_solution_step6}
  V(t)\leq \gamma e^{-C_0 t}-\gamma^5 A_2 w_{\varepsilon,\kappa,d}(t)e^{-\frac{\varepsilon}{\kappa}t}\bm{1}_{\{ t\geq 2t_{\gamma} \}}
\end{equation}
for $t\leq T$. And
\begin{equation}
  \label{any_solution_step7}
  V(T)=\gamma e^{-C_{\gamma}T}-\gamma^3 A_1 w_{\varepsilon,\kappa,d}(T)e^{-\frac{\varepsilon}{2\kappa}T}
\end{equation}
or
\begin{equation}
  \label{any_solution_step8}
  V(T)=\gamma e^{-C_0 T}-\gamma^5 A_2 w_{\varepsilon,\kappa,d}(T)e^{-\frac{\varepsilon}{\kappa}T}\bm{1}_{\{ T\geq 2t_{\gamma} \}}.
\end{equation}
Suppose first that $T\leq T'$. In particular, $V$ is decreasing on the interval $[0,T]$. Now, all the arguments leading to ineqs.~\eqref{VWinK_step1} and~\eqref{VWinK_step2} can be repeated (using ineqs.~\eqref{any_solution_step5}, \eqref{any_solution_step6} and the monotonicity of $V$ on the interval $[0,T]$) to show that
\begin{equation}
  \label{any_solution_step9}
  V(t)>\gamma e^{-C_{\gamma}t}-\gamma^3 A_1 w_{\varepsilon,\kappa,d}(t)e^{-\frac{\varepsilon}{2\kappa}t}
\end{equation}
and
\begin{equation}
  \label{any_solution_step10}
  V(t)<\gamma e^{-C_0 t}-\gamma^5 A_2 w_{\varepsilon,\kappa,d}(t)e^{-\frac{\varepsilon}{\kappa}t}\bm{1}_{\{ t\geq 2t_{\gamma} \}}
\end{equation}
for $t\leq T$. This contradicts eqs.~\eqref{any_solution_step7} or \eqref{any_solution_step8}; therefore, $T\geq T'$. Again, all the arguments leading to ineq.~\eqref{VWinK_step4} can be repeated (using ineqs.~\eqref{any_solution_step5}, \eqref{any_solution_step6} and the monotonicity of $V$ on the interval $[0,T']$) to show that
\begin{equation}
  \label{any_solution_step11}
  \frac{dV(t)}{dt}<-\frac{C_0}{2}\gamma^2
\end{equation}
for $t\leq t_{\gamma}\cap T'$. This and eq.~\eqref{any_solution_step4} imply $T'\geq t_{\gamma}$. This shows that $V$ is decreasing on the interval $[0,t_{\gamma}]$. Then, again, all the arguments leading to ineqs.~\eqref{VWinK_step1} and~\eqref{VWinK_step2} can be repeated (using ineqs.~\eqref{any_solution_step5}, \eqref{any_solution_step6} and the monotonicity of $V$ on the interval $[0,t_{\gamma}]$) to show that ineqs.~\eqref{any_solution_step9} and~\eqref{any_solution_step10} hold for $t\leq T$. This contradicts eqs.~\eqref{any_solution_step7} or \eqref{any_solution_step8}: $T=+\infty$. And repeating the argument above shows $T'\geq t_{\gamma}$; therefore, $(f,V)$ satisfies ineqs.~\eqref{lower_thm1}, \eqref{upper_thm1} and $V$ is decreasing on the interval $[0,t_{\gamma}]$. This completes the proof of Theorem~\ref{thm:large_kappa}.
\section{Discussion} 
\label{sec:discussion}
Theorem~\ref{thm:large_kappa} gives a theoretical basis of the numerical observation given in~\cite{Tsuji2012}: The interaction of the molecules with the dispersed obstacles causes exponential decay of the velocity $V(t)$ if $\varepsilon>0$; if, on the other hand, $\varepsilon=0$, then $V(t)$ decays algebraically with a rate independent of the spatial dimension $d$.

The proof reveals the mathematical structure determining the long time behavior of $V(t)$. In particular, it explains why the algebraic decay rate in the case of $\varepsilon=0$ is independent of $d$: In the integral $J$ (eq.~\eqref{J}), $\hat{\xi}$ satisfying
\begin{equation}
  \label{xi_hat_major_contribution}
  \nu_{\varepsilon}(|\hat{\xi}|)=O(\kappa/t) \quad (t\to \infty)
\end{equation}
gives the major contribution. And since $\nu_{0}(z)=\pi^{1/2}z/2$, this implies
\begin{equation}
  \label{xi_hat_restriction}
  |\hat{\xi}|=O(\kappa/t) \quad (t\to \infty).
\end{equation}
Roughly speaking, this means that the length of the rigid body is effectively finite also in the $\hat{\xi}$-direction (see ineq.~\eqref{xiperp_estimate} in Lemma~\ref{lem:xi_estimates}); therefore, dimensional dependence disappears.

Next, I discuss some variants of Theorem~\ref{thm:large_kappa}.

As remarked below Theorem~\ref{thm:large_kappa} (Remark~\ref{rem:thm1} (iv)), the long time behavior of $V(t)$ is qualitatively different if $\varepsilon \geq 2\kappa C_0$:

\begin{theorem}
  \label{thm:small_kappa}
  Suppose that $\kappa \geq 1$ and $\varepsilon \geq 2\kappa C_0$. Then for $\gamma>0$ sufficiently small, there exists a solution $(f,V)$ to eqs.~\eqref{Lorentz}, \eqref{init_gas} and~\eqref{specular}; and eq.~\eqref{Newton} with $V_0=\gamma$ satisfying the following inequalities:
  \begin{equation} 
    \label{small_kappa_bounds}
    \frac{1}{2}\gamma e^{-C_0 t}\leq V(t)\leq \gamma e^{-C_0 t}.
  \end{equation}
  Moreover, any solution $(f,V)$ satisfies these inequalities and $V$ is decreasing on the whole interval $[0,\infty)$.
\end{theorem}

Theorem~\ref{thm:small_kappa} implies that $V(t)$ is always positive, which is in contrast with the occurrence of sign change in Theorem~\ref{thm:large_kappa} (Remark~\ref{rem:thm1} (iii)). The basic strategy of the proof is the same as that of Theorem~\ref{thm:large_kappa}. The proof is given in the appendix.

Theorem~\ref{thm:large_kappa} can be extended to the case with a constant external force $E>0$: that is, eq.~\eqref{Newton} is replaced by
\begin{equation}
  \label{Newton_E}
  \frac{d}{dt}V(t)=E-D(t), \quad V(0)=V_0.
\end{equation}
Note that in the free molecular case, this problem was considered in~\cite{Caprino2007,Caprino2006}: It was shown that the velocity $V(t)$ approaches the terminal velocity $V_{\infty}=D_{0}^{-1}(E)$ algebraically as $V_{\infty}-V(t)\approx t^{-(d+2)}$. In the Lorentz gas case, the approach becomes exponential (even if $\varepsilon=0$). The proof is similar to that of Theorem~\ref{thm:large_kappa} --- much easier because $r_{W}^{\pm}(t)\geq 0$ becomes trivial in this case.

There are more variants studied in the free molecular case: with linear restoring force~\cite{Caprino2007}, other rigid body shapes~\cite{Cavallaro2007,Sisti2014,Fanelli2016}, the rigid body replaced by an elastic body~\cite{Cavallaro2012} and when the gas fills the half-space~\cite{Koike2017}. These all assume the specular boundary condition; other boundary conditions including the Maxwell boundary condition were studied in~\cite{Aoki2008,Chen2014,Chen2015}. It is reasonable to expect that Theorem~\ref{thm:large_kappa} can be extended to these variants. Also, the case with linear restoring force under the diffuse boundary condition treated numerically in~\cite{Tsuji2012} (both the free molecular and Lorentz gas cases) may be handled mathematically using the techniques developed in this paper and the references above; I have not, however, examined these cases in detail.

Lastly, I briefly comment on the case of the Boltzmann equation. Tsuji and Aoki~\cite{Tsuji2013} considered this case numerically (with linear restoring force under the diffuse boundary condition in the spatially one dimensional case $d=1$) and observed that $V(t)$ decays algebraically: $V(t)\approx t^{-3/2}$. Note that the decay rate $-3/2$ is slower than $-2$ in the free molecular case~\cite{Tsuji2012}. This result is in sharp contrast with the exponential decay in the Lorentz gas case. As the analysis in this paper suggests, intermolecular collisions destroy memory effect due to recollision; therefore, I suspect that it is due to more fluid like effect as discussed in~\cite{Belmonte2001,Vazquez2003,Cavallaro2011,Butta2015}. Mathematical understanding of the memory effect in rarefied gases needs further investigation.

\appendix
\section*{Appendix: Proof of Theorem~\ref{thm:small_kappa}}
\label{appendix:proof_thm:small_kappa}
The basic strategy is the same as that of Theorem~\ref{thm:large_kappa}: (i) Define an appropriate function space $\mathcal{K}$. (ii) Prove decay estimates of $r_{W}^{\pm}(t)$ given $W\in \mathcal{K}$. (iii) Prove that $V_W$ again belongs to $\mathcal{K}$. (iv) Apply Schauder's fixed point theorem to show that the map $\mathcal{K}\ni W\mapsto V_W \in \mathcal{K}$ has a fixed point $V\in \mathcal{K}$ --- this shows the existence part. (v) Prove that any solution $(f,V)$ satisfies ineqs.~\eqref{small_kappa_bounds} and that $V$ is decreasing on the interval $[0,\infty)$.

Only steps (i) through (iii) are explained in this section; Steps (iv) and (v) can be carried out similarly as in the proof of Theorem~\ref{thm:large_kappa}.

Step (i) is the following:

\begin{definition}
  \label{def:space2}
  Let $\gamma>0$. A Lipschitz continuous function $W\colon [0,\infty)\to \mathbb{R}$ belongs to $\mathcal{K}=\mathcal{K}(\gamma)$ if $W(0)=\gamma$; $W$ is decreasing on the interval $[0,\infty)$; and satisfies
  \begin{equation}
    \label{space2_bounds}
    \frac{1}{2}\gamma e^{-C_0 t}\leq W(t)\leq \gamma e^{-C_0 t}
  \end{equation}
  and $|dW(t)/dt|\leq 1$.
\end{definition}

Let $W\in \mathcal{K}$. Note that $r_{W}^{\pm}(t)\geq 0$ since $W(t)>0$. Moreover, $r_{W}^{+}(t)=0$ since $W$ is decreasing on the interval $[0,\infty)$; therefore, only $r_{W}^{-}(t)$ needs analysis.

Step (ii) is to prove the following:

\begin{proposition}
  \label{prop:rWm_upper_small_kappa}
  Let $\kappa \geq 1$ and $\varepsilon \geq 2\kappa C_0$. If $W\in \mathcal{K}(\gamma)$, then
  \begin{equation}
    \label{rWm_upper_small_kappa}
    0\leq r_{W}^{-}(t)\leq C\gamma^3 \frac{e^{-C_0 t}}{(1+t)^{d+2}}
  \end{equation}
  for some positive constant $C$ independent of $\gamma$.
\end{proposition}

\begin{proof}
  Define $L_t$ by eq.~\eqref{Lt} and let
  \begin{equation}
    \label{Lt*Lt**}
    L_{t}^{*}=\{ (x,\xi)\in L_t \mid \tau_1 \leq t/2 \}, \quad L_{t}^{**}=\{ (x,\xi)\in L_t \mid \tau_1>t/2 \}.
  \end{equation}
  Then $r_{W}^{-}(t)=r_{W,*}^{-}(t)+r_{W,**}^{-}(t)$, where
  \begin{align}
    r_{W,*}^{-}(t) & =2\int_{L_{t}^{*}}(\xi_1-W(t))^2 (f_0-f_W)\, d\xi dS, \label{rW*m} \\
    r_{W,**}^{-}(t) & =2\int_{L_{t}^{**}}(\xi_1-W(t))^2 (f_0-f_W)\, d\xi dS. \label{rW**m}
  \end{align}
  Note first that $(x,\xi)\in L_t$ implies
  \begin{equation}
    \label{xi1-W_Lt}
    0<\xi_1 -W(t)<\xi_1 \leq \frac{1}{t-\tau_1}\int_{\tau_1}^{t}\gamma e^{-C_0 s}\, ds=\gamma e^{-C_0 \tau_1}\frac{1-e^{-C_0(t-\tau_1)}}{C_0(t-\tau_1)}.
  \end{equation}
  Similar arguments as in the proof of Proposition~\ref{prop:|rWm|_upper} show
  \begin{align}
    r_{W,*}^{-}(t) & \leq C\gamma^3 w_{\varepsilon,\kappa,d}(t)e^{-\frac{\varepsilon}{2\kappa}t}, \label{rW*m_upper} \\
    r_{W,**}^{-}(t) & \leq C\gamma^3 e^{-\frac{3}{2}C_0 t}. \label{rW**m_upper}
  \end{align}
  Since $\varepsilon/2\kappa \geq C_0$, these imply ineq.~\eqref{rWm_upper_small_kappa}. \qed
\end{proof}

This is step (ii). Step (iii) is to prove the following:

\begin{proposition}
  \label{prop:VWinK_small_kappa}
  Let $\kappa \geq 1$ and $\varepsilon \geq 2\kappa C_0$. Then $W\in \mathcal{K}=\mathcal{K}(\gamma)$ implies $V_W \in \mathcal{K}$ for sufficiently small $\gamma$.
\end{proposition}

\begin{proof}
  Since $r_{W}^{\pm}(t)\geq 0$, the upper bound
  \begin{equation}
    \label{VWinK_small_kappa_step1}
    V_W(t)\leq \gamma e^{-C_0 t}
  \end{equation}
  follows from eq.~\eqref{VW_formula}.

  To prove the lower bound
  \begin{equation}
    \label{VWinK_small_kappa_step2}
    V_W(t)\geq \frac{1}{2}\gamma e^{-C_0 t},
  \end{equation}
  I need first to show
  \begin{equation}
    \label{VWinK_small_kappa_step3}
    e^{-\int_{0}^{t}K(W(s))\, ds}\geq \frac{3}{4}e^{-C_0 t}
  \end{equation}
  for sufficiently small $\gamma$: Note first that $D_{0}''(0)=0$ and $D_{0}'''(0)>0$ by eq.~\eqref{D0}; hence for $\gamma \leq 1$,
  \begin{align}
    \label{VWinK_small_kappa_step4}
    \begin{aligned}
      0<K(W(s)) & =D_0(W(s))/W(s) \\
                & \leq D_{0}'(W(s)) \\
                & \leq C_0+C(W(s))^2 \\
                & \leq C_0+C\gamma^2 e^{-2C_0 s}
    \end{aligned}
  \end{align}
  for some positive constant $C$. This implies
  \begin{equation}
    \label{VWinK_small_kappa_step5}
    e^{-\int_{0}^{t}K(W(s))\, ds}\geq e^{-\int_{0}^{t}(C_0+C\gamma^2 e^{-2C_0 s})\, ds}=e^{-C\gamma^2}e^{-C_0 t}.
  \end{equation}
  And taking $\gamma$ sufficiently small leads to ineq.~\eqref{VWinK_small_kappa_step3}.

  Now ineq.~\eqref{VWinK_small_kappa_step2} follows from eq.~\eqref{VW_formula}, Proposition~\ref{prop:rWm_upper_small_kappa} and ineq.\eqref{VWinK_small_kappa_step3} as follows:
  \begin{align}
    \label{VWinK_small_kappa_step6}
    \begin{aligned}
      V_W(t)
      & \geq \frac{3}{4}\gamma e^{-C_0 t}-C\gamma^3 \int_{0}^{t}e^{-C_0(t-s)}\frac{e^{-C_0 s}}{(1+s)^{d+2}}\, ds \\
      & =\frac{3}{4}\gamma e^{-C_0 t}-C\gamma^3 e^{-C_0 t}\int_{0}^{t}\frac{ds}{(1+s)^{d+2}} \\
      & \geq \frac{1}{2}\gamma e^{-C_0 t}
    \end{aligned}
  \end{align}
  if $\gamma$ is sufficiently small.

  Next, I show that $V_W$ is decreasing on the interval $[0,\infty)$: Since $V_W>0$ (by ineq.~\eqref{VWinK_small_kappa_step2}) and $r_{W}^{\pm}(t)\geq 0$,
  \begin{equation}
    \label{VWinK_small_kappa_step7}
    \frac{dV_W(t)}{dt}=-K(W(t))V_W(t)-r_{W}^{+}(t)-r_{W}^{-}(t)<0
  \end{equation}
  for $t\geq 0$.

  Lastly, since $r_{W}^{+}(t)=0$, $|r_{W}^{-}(t)|\leq C\gamma^3$ (by Proposition~\ref{prop:rWm_upper_small_kappa}) and $|V_W(t)|\leq \gamma$ (by ineqs.~\eqref{VWinK_small_kappa_step1} and \eqref{VWinK_small_kappa_step2}),
  \begin{equation}
    \label{VWinK_small_kappa_step8}
    \left| \frac{dV_W(t)}{dt} \right|\leq K(W(t))|V_W(t)|+|r_{W}^{+}(t)|+|r_{W}^{-}(t)|\leq C_{\gamma}\gamma+C\gamma^3 \leq 1
  \end{equation}
  for sufficiently small $\gamma$. \qed
\end{proof}

Steps (iv) and (v) can be carried out similarly to the proof of Theorem~\ref{thm:large_kappa}. And Theorem~\ref{thm:small_kappa} is proved.


\bibliographystyle{spmpsci}
\bibliography{bibtex_rigidbodymotion}

\end{document}